\documentclass[11pt]{article}
\usepackage{amsmath,amsfonts,amssymb,bbm,amsthm, bm}
\usepackage[numbers]{natbib} 
\usepackage{xcolor}
\usepackage[letterpaper, margin=1in]{geometry}
\usepackage{thm-restate}

\usepackage{hyperref}
\usepackage{cleveref}
\usepackage{float}
\usepackage{graphicx}
\usepackage{enumitem}
\usepackage{caption}

\usepackage{algorithm,algorithmic}

\usepackage{hhline}

\usepackage{booktabs} 

\newtheorem{theorem}{Theorem}[section]
\newtheorem{lemma}[theorem]{Lemma}

\newtheorem{prop}[theorem]{Proposition}

\newtheorem{corollary}[theorem]{Corollary}

\newtheorem{assumption}{Assumption}

\theoremstyle{definition}
\newtheorem{definition}[theorem]{Definition}
\newtheorem{remark}[theorem]{Remark}
\newtheorem*{example}{Example}

\DeclareMathOperator*{\E}{{\mathbb{E}}}

\newcommand{\N}[0]{{\mathbb{N} }}

\newcommand{\R}[0]{{\mathbb{R} }}
\newcommand{\Ratio}[0]{joint to marginal product ratio }
\DeclareMathOperator{\ratio}{JP}

\newcommand{\X}[0]{{\mathcal{X} }}
\newcommand{\Y}[0]{{\mathcal{Y} }}
\newcommand{\Z}[0]{{\mathcal{Z} }}

\DeclareMathOperator{\dom}{dom}

\newcommand{\defn}[1]{{\textbf{\textit{#1}}}}

\title{Learning and Strongly Truthful Multi-Task Peer Prediction: A Variational Approach\footnote{Grant Schoenebeck and Fang-Yi Yu are
pleased to acknowledge the support of the National Science Foundation NSF 1618187 and 2007256}}
\author{
  Grant Schoenebeck\thanks{
  University of Michigan,
  \texttt{schoeneb@umich.edu}} \and
  Fang-Yi Yu\thanks{Harvard University \texttt{fayu@umich.edu}}
}
\begin{document}
\maketitle
\begin{abstract}
Peer prediction mechanisms incentivize agents to truthfully report their signals even in the absence of verification by comparing agents' reports with those of their peers. In the detail-free multi-task setting, agents are asked to respond to multiple independent and identically distributed tasks, and the mechanism does not know the prior distribution of agents' signals. The goal is to provide an $\epsilon$-strongly truthful mechanism where truth-telling rewards agents ``strictly'' more than any other strategy profile (with $\epsilon$ additive error) even for heterogeneous agents, and to do so while requiring as few tasks as possible.

We design a family of mechanisms with a scoring function that maps a pair of reports to a score.  The mechanism is strongly truthful if the scoring function is ``prior ideal."  Moreover, the mechanism is $\epsilon$-strongly truthful as long as the scoring function used is sufficiently close to the ideal scoring function.  This reduces the above mechanism design problem to a learning problem--- specifically learning an ideal scoring function.  Because learning the prior distribution is sufficient (but not necessary) to learn the scoring function, we can apply standard learning theory techniques that leverage side information about the prior (e.g., that it is close to some parametric model). Furthermore, we derive a variational representation of an ideal scoring function and reduce the learning problem into an empirical risk minimization.  

We leverage this reduction to obtain very general results for peer prediction in the multi-task setting.  Specifically,
\begin{description}
\item[Sample Complexity]  We show how to derive good bounds on the number of tasks required for different types of priors--in some cases exponentially improving previous results.  In particular, we can upper bound the required number of tasks for parametric models with bounded learning complexity.  Furthermore, our reduction applies to myriad continuous signal space settings.  To the best of our knowledge, this is the first peer-prediction mechanism on continuous signals designed for the multi-tasks setting.      
\item[Connection to Machine Learning]  We show how to turn a soft-predictor of an agent's signals (given the other agents' signals) into a mechanism.  This allows the practical use of machine learning algorithms that give good results even when many agents provide noisy information.  
\item[Stronger Properties]  Our mechanisms apply to any stochastically relevant prior rather than the more restrictive settings of previous mechanisms.  In the finite setting, we obtain $\epsilon$-strongly truthful mechanisms, whereas prior work only achieves a weaker notion of truthfulness (informed truthfulness)~\cite{Shnayder2016-xx, Agarwal2017-ty}.
\end{description}
\end{abstract}

\newpage
\setcounter{page}{1}

\section{Introduction}
Peer prediction is the problem of information elicitation without verification.  Peer prediction mechanisms exploit the interdependence in agents' signals to incentive agents to report their private signal truthfully even when the reports cannot be directly verified.  In \emph{the multi-task setting}~\cite{dasgupta2013crowdsourced}, each agent is asked to respond to multiple, independent tasks.  For example:

\begin{example}[Commute time] We can collect data from drivers to estimate the commute time of a certain route.  Each driver's daily commute time might be modeled in the following way: each day, the route has an expected time generated from a Gaussian distribution, and each driver's commute time is the expected time perturbed by independently distributed Gaussian noise.
\end{example}

Peer prediction from strategic agents has been attracting a surge of interest in economics and computer science.  Several previous works~\cite{Agarwal2017-ty, kong2019information,liuchen} can be understood as using particular learning algorithms to learn nice payment functions that capture the interdependence in agents' reports.  In this paper, we decouple these two components: mechanism design and learning algorithms.  This framework provides a clean black-box reduction from learning algorithms to peer prediction mechanism.

One advantage of our framework is that we can use results from machine learning about complexity of learning parameters of  priors to obtain bounds on the sample complexity (number of tasks required) of our mechanism.  For instance,  using our reduction, we can easily exponentially improve the required number of tasks in the previous work~\cite{Shnayder2016-xx}.  

Two features of our mechanisms enable us to work in more complicated settings.   First, our mechanisms use mutual information to pay agents.  This allows us to use aggregation algorithms and pay an agent the mutual information between her reports and the aggregated outcome of the other agents.   For example, suppose the agents' report's average quality is low, and a large fraction of agents report random noise.  In that case, we can use aggregation to enhance the signal to noise ratio and provide a robust incentive to strategic workers.  
The second feature of our mechanisms is a variational formulation, which ensures one-sided error such that we can only underestimate the mutual information but not overestimate it.  This allows us to measure different scoring functions' accuracy agnostically. Thus, we can use deep learners or other rich enough functions to learn a good payment in practice. 

In addition to the above contributions, we also improve previous work in two axes: the \emph{truthfulness guarantee} and the \emph{prior assumption}.  

The truthful guarantee explains \emph{how good the truth-telling strategy is in the mechanism} (formally defined in Sect.\ref{sec:goals}). Is the truth-telling always the best response regardless of other's strategy (dominantly truthful)? Or truth-telling is a Bayesian Nash equilibrium, and agents get strictly higher payment than any other non-permutation equilibrium (strongly truthful) where a permutation equilibrium is one where agents report a permutation of the signals.  A slightly weaker property is \emph{informed truthful} where no strategy profile pays strictly more than truth-telling, and truth-telling pays more than any uninformative equilibrium.  Our pairing mechanisms is dominantly truthful if the number of tasks is infinite, and approximately strongly truthful when the number of tasks is finite.

Previous peer prediction mechanisms make ad hoc assumptions on agents' private signals (positively correlated~\cite{dasgupta2013crowdsourced}, fine-grained~\cite{kong2019information}, strictly correlated~\cite{kong2020dominantly}, or latent variable models~\cite{liuchen}) which are discussed in Sect.~\ref{sec:prior}.  Moreover, all the above mechanisms can only work when agents' signals are in a finite space.
Under these assumptions, a question that bears asking is how generic their truthful guarantee is.  At one extreme, if agents' private signal is always from a single known distribution, it is trivial to design a strongly truthful mechanism. Therefore, another axis to measure a peer prediction mechanism's performance is its prior assumption, which tells \emph{how general agents' prior can be}.  

There are two motivations to understand how general agents' prior can be.  First, in practice, we need a peer prediction mechanism that works for continuous signals e.g., the above Commute time example, but the previous mechanisms cannot.\footnote{Discretization approach is not practical in most situations~\cite{kong2020information}.}  Second, a mechanism's prior assumption often reveals why the mechanism works.  Thus, improving prior assumptions can push our theoretical understanding of peer prediction mechanisms.  It is well-known that to have the truth-telling strategy profile as a strict Bayesian Nash equilibrium, one necessary condition is that agents' signals need to be \emph{stochastic relevant}  (Definition~\ref{def:nondegenerate})~\cite{zhang2014elicitability}.  However, when is stochastic relevance a sufficient condition?  
In this paper, we show stochastic relevance is also a sufficient condition in the multi-tasks setting.  Our pairing mechanisms achieve approximately-strongly truthful as long as the prior is stochastic relevant.  In particular, the space of agents' signals can be countably infinite or even continuous.  To the authors' knowledge, our mechanism is the first that works on the maximal possible prior structures in the multi-task setting.

Besides the above properties, we also require our mechanisms 1) are \emph{minimal} which only elicit the agents' signals and no additional information; 2) are \emph{detail-free} which do not require foreknowledge of the prior; and 3) have \emph{low sample number}, where each agent only needs to answer a few questions for the mechanism to achieve approximately strong truthfulness. (Definition~\ref{def:approx_strongly}).

\paragraph{Our Techniques:}
Prior work~\cite{kong2019information} has shown that paying agents according to the  $\Phi$ mutual information (a generalization of the Shannon mutual information) between their signals is a good idea.  This is because, if agents try to strategically manipulate their signals, the $\Phi$ mutual information can only decrease.  However, a key open question is how to compute the mutual information while having access to only a few signals for each agent.  Moreover, the computation needs to be done in a way that maintains the incentive guarantees of the mechanism. 

We solve this issue.  First, we convert the mechanism design problem into an optimization problem (Theorem~\ref{thm:framework}).  The $\Phi$ mutual information of a pair of random variables can be defined as the $\Phi$ divergence between two distributions: the joint distribution and the product of marginal distributions.  The $\Phi$ divergence is just a measure of distance between the two distributions and contains the KL-divergence as a special case.  The problem of computing the  $\Phi$ divergence, using variational representation as a bridge, can be changed into the optimization problem of finding the best ``distinguisher'' between these two distributions.  We call such a distinguisher a \emph{scoring function}.  The optimal scoring function  (distinguisher) can differentiate the two distributions with a score equal to the $\Phi$ divergence, whereas any other  scoring function (distinguisher) yields a lower score.   Thus, once one has this optimal scoring function, estimating the  $\Phi$ divergence (and hence  $\Phi$ mutual information) is easy--just compute its score.  In this paper we call the optimal scoring function for a particular prior $P$, the \emph{$(P, \Phi)$-ideal  scoring function} which can be easily computed when the prior $P$ is known.

Our mechanism will reward agents according to some scoring function.  Importantly, agents' ex-ante payments under prior $P$ are maximized when \emph{both} the distinguisher used is the $(P, \Phi)$-ideal scoring function, and the agents are truth-telling.  Consequently, if we already have the $(P, \Phi)$-ideal scoring function, the mechanism incentivizes truthful reporting. Furthermore, agents are willing to help the mechanism to learn the  $(P, \Phi)$-ideal scoring function rather than to trick it into using a suboptimal scoring function.  

Compared with \citet{kong2019information}, our variational characterization provides a better truthfulness guarantee when the number of tasks is finite.  We can uniformly upper bound the ex-ante payments under any non-truthful strategy profile (Definition~\ref{def:approx_strongly}) even when the learning algorithm cannot estimate the ideal scoring functions under those non-truthful strategies.   This property is vital for continuous signal spaces where agents may adversarially adopt the worst possible strategy profiles to compromise the learning algorithm.

The above observations transform the problem from designing a mechanism to simply learning the $(P, \Phi)$-ideal scoring function given samples from a prior.  We provide two algorithms to learn the scoring function.  The first one is a \emph{generative} approach which estimates the whole density function of the prior and computes a scoring function from it. In a \emph{discriminative} approach, we formulate the estimation of the ideal scoring function as a convex optimization problem, empirical risk minimization~\cite{nguyen2010estimating}, and estimate the scoring function directly.  This latter approach allows us to use state-of-art convex optimization solvers to estimate good scoring functions.

\paragraph{Our Contributions:}
In this paper, we leverage the above insights to design a $\Phi$-pairing mechanism that is minimal and detail-free for heterogeneous agents.  In particular:
\begin{description}
\item[Sample Complexity]  We show how to derive good bounds on the number of tasks required for different types of priors--in some cases exponentially improving previous results.  In particular, we can upper bound the required number of tasks for parametric models with bounded learning complexity (as measured by a continuous analog of the VC dimension).   Furthermore, our reduction applies to myriad continuous signal space settings.  To the best of our knowledge, this is the first peer-prediction mechanism on continuous signals designed for the multi-question setting.      
\item[Connection to Machine Learning]  We show how to turn a soft-predictor of an agent's signals (given the other agents' signals) into a mechanism.  This allows the practical use of machine learning algorithms that give good results even when many agents provide noisy information.  
\item[Stronger Properties]  Our mechanisms apply to any stochastically relevant prior rather than the more restrictive settings of previous mechanisms.  In the finite setting, we obtain $\epsilon$-strongly truthful mechanisms, whereas prior work only achieves a weaker notion of truthfulness (informed truthfulness) ~\cite{Shnayder2016-xx, Agarwal2017-ty}.
\end{description}
 \begin{tabular}{l  c  c  c c c}
 \toprule
  & D\&G~\citep{dasgupta2013crowdsourced}  & CA~\citep{Shnayder2016-xx,Agarwal2017-ty}  &  $\Phi$-MIM~\citep{kong2019information} & DMI~\citep{kong2020dominantly} & \vtop{\hbox{\strut $\Phi$-pairing }\hbox{\strut mechanism}} \\
  
 \hline
 Signal space        &   binary   &   finite   & finite &   finite &  continuous \\ \hline
 \vtop{\hbox{\strut Prior }\hbox{\strut Assumptions}} & \vtop{\hbox{\strut positive }\hbox{\strut  correlated}}  & \vtop{\hbox{\strut stochastic }\hbox{\strut  relevant}} & \vtop{\hbox{\strut fine}\hbox{\strut -grained}} & \vtop{\hbox{\strut strictly }\hbox{\strut  correlated}} &  \vtop{\hbox{\strut stochastic }\hbox{\strut  relevant}} \\ \hline
 Truthful           &   \checkmark   &   \checkmark   &   \checkmark&   \checkmark & \checkmark\\ \hline
 Informed-truthful  &   \checkmark   &    \checkmark      &   \checkmark &  \checkmark &  \checkmark\\ \hline
Strongly truthful   &   \checkmark   &       &   \checkmark(fine-grained) &  \checkmark &  \checkmark \\ \hline
  Detail-free       &   \checkmark   & \checkmark    &   \checkmark  & \checkmark & \checkmark\\ \hline
 Samples     &   &  $O(n)$    &  $\infty$ &  $O(|\Omega|^2)$ &  $O(\log n)$ \\
 \bottomrule
\end{tabular}

In the above table, $\Omega$ is the signal space required to be shared by all agents.

\subsection{Related Work}
\paragraph{Multi-task setting}
In the multi-task setting,  \citet{dasgupta2013crowdsourced} propose a \emph{strongly truthful} mechanism when the signal space is binary and every pair of agents' signals are assumed to be positively correlated.  
Both \citet{kong2019information} and \citet{Shnayder2016-xx} independently generalize \citet{dasgupta2013crowdsourced} to discrete signal spaces, though in  different manners illustrated as follows.   

\citet{kong2019information} present the \defn{$\Phi$-mutual information mechanism}, a multi-task peer prediction mechanism for the finite signal space setting with arbitrary interdependence between signals.  Unfortunately, the sample number is infinite.  They show that their mechanism is  strongly truthful as long as the prior is ``fine-grained" (it is truthful in any event).  A prior is \emph{fine-grained} if, roughly speaking, no two signals can be interpreted as different names for the same signal.
To define their mechanism they introduce the notion of $\Phi$-mutual information (of which Shannon mutual information is a special case) where $\Phi$ is any convex function.  
Their mechanism pays each agent the $\Phi$-mutual information  between her reports and the reports of another randomly chosen agent.  Strategic behavior is shown to not increase  $\Phi$-mutual information by a generalized version of the data processing inequality.  Unfortunately, their analysis requires infinite sample number to measure this $\Phi$-mutual information and does  not handle errors in estimation.

\citet{Shnayder2016-xx} introduce the \defn{Correlated Agreement (CA) mechanism} which also generalizes~\citet{dasgupta2013crowdsourced} to any finite signal space.  On the one hand, the CA mechanism can assume the knowledge of the ``signal structure'' (which tells which signals are positively and negatively correlated).  In this case they can provide a mechanism that is truthful with sample number of two.\footnote{The original paper shows it requires $3$, but it actually only needs $2$ tasks.} On the other hand, when agents are homogeneous the CA mechanism can learn the signal structure, albeit with some chance of error, if it has sample number $O(n)$.  The CA mechanism is shown to be robust to this error, and is $\epsilon$-truthful.  In both cases the CA mechanism is actually $\epsilon$-\emph{informed truthful} (a slightly weaker notion than strongly truthful).
\citet{Agarwal2017-ty} extend the above work of \citet{Shnayder2016-xx} to a particular setting of heterogeneous agents where agents are (close to) one of a fixed number of types.  They again establish a $O(n)$ sample number in this new setting.

Note that in the above works, a new robustness (error) analysis is required for each different setting of interdependence between signals.  Interestingly, the CA mechanism can be viewed as a special case of the aforementioned $\Phi$-mutual information mechanism using the total variation distance mutual information (i.e., $\Phi(a) = |a-1|/2$).  However, instead of directly computing this mutual information, the CA mechanism obtains a consistent estimator of it~\cite{kong2019information}.  Similarly, in the special case that our mechanism implements the total variation distance, we also recover the CA mechanism.  However, our analysis is entirely different. 

\citet{kong2020dominantly} shows an elegant way of obtaining strongly truthful mechanisms (DMI mechanism) for the multitask setting.  Our results are incommensurate with these results.  In our results, the sample complexity grows with the $\epsilon$ in the desired $\epsilon$-strongly truthful guarantee but is independent of the number of signals.  In~\citet{kong2020dominantly}, there is an exact strongly truthful guarantee with sample complexity grows in the size of the signal space.  However, the prior structure needs to be strictly correlated, which is a stronger assumption on stochastic relevance.  We provide comparison at Sect.~\ref{sec:prior}.  In particular, her mechanism requires all agents' report space are all finite and have the same size.  This restricts applications of the aggregation algorithm mentioned in the introduction and Sect.~\ref{sec:more}.

\paragraph{Single task setting} In general, agents do not (necessarily) have multiple identical and independent signals.   Without this property, most of the mechanisms require knowledge of a common prior (not detail-free) or for agents to report their whole posterior distribution of other's signals (not minimal).  The later solution is especially difficult to apply to complicated signal spaces (e.g. asking agents to report their probability density function of others' continuous signals). 

\citet{MRZ05} introduce the peer prediction mechanism which is the first mechanism that has truth-telling as a strict Bayesian Nash equilibrium and does not need verification. However, their mechanism requires the full knowledge of the common prior and there exist some equilibria that are paid more than truth-telling. In particular, the oblivious equilibrium pays strictly more than truth-telling. \citet{2016arXiv160307319K} modify the original peer prediction mechanism such that truth-telling pays strictly better than any other equilibrium but still requires the full knowledge of the common prior. \citet{prelec2004bayesian} designs the first detail-free peer prediction mechanism---Bayesian truth serum (BTS) in the one quesetion setting. Several other works study the one-question setting of BTS \cite{radanovic2013robust,radanovic2014incentives,witkowski2012peer,kong2016equilibrium}.  For continuous signals, \citet{radanovic2014incentives} apply a discretization approach and use a new payment method, but that is also non-minimal.   \citet{goelpersonalized} work on a mixture of normal distributions with an infinite number of agents.

\paragraph{Miscellany} 
\citet{Liu:2017:MAP:3033274.3085126} design a peer prediction mechanism where each agents' responses are not compared to another agents', but rather the output of a machine learning classifier that learns from all the other agents' responses.  \citet{liuchen} design a non-minimal approximate dominant strategy mechanism that uses surrogate loss functions as tools to correct for the mistakes in agents' reports. \citet{kong2018water} studies the related goal for forecast elicitation, and like the present work uses Fenchel's duality to reward truth-telling (though in a different manner).

One interesting, but orthogonal, line of work looks at ``cheap" signals, where agents can coordinate on less useful information.   For example, instead of grading an assignment based on correctness, a grader could only spot check the grammar.    \citet{gao2016incentivizing} introduces the issue, while \citet{kong2018eliciting} shows a partial solution using conditional mutual information.  

The recent book~\citet{faltings2017game} surveys additional results from this area.

\subsection{Structure of Paper}  
Sect.~\ref{sec:pre} introduces some basic notions in this paper.  In particular, Sect.~\ref{sec:prior} defines scoring functions, which will play an important role in this paper.

At the beginning of Sect.~\ref{sec:pairing}, we define a central component of our $\Phi$-pairing mechanism, Mechanism~\ref{alg:fmechansim}, which takes agents' report and a scoring function $K$ as input.  In Sect.~\ref{sec:known_prior}, we consider the full information setting.  We show, in the Mechanism~\ref{alg:fmechansim} with an ideal scoring function, agents are incentivized to report their signals truthfully.  In Sect.~\ref{sec:lem}, we prove Theorem~\ref{thm:truth}, and main technical lemmas.

In Sect.~\ref{sec:free}, we define a notion of approximation of an ideal scoring function and introduce our framework that reduces the mechanism problem for information elicitation to a learning problem for an ideal scoring function (Theorem~\ref{thm:framework}).  

In Sect.~\ref{sec:learning}, we focus on the learning problem introduced in Sect.~\ref{sec:free}.  We first show two sufficient conditions for approximating an ideal scoring function in Sect.~\ref{sec:sufficient}.  Then, we present two algorithms to derive approximately ideal scoring functions from agents' reports in Sect.~\ref{sec:algorithm}.  Additionally, in Sect.~\ref{sec:nonexistence}, we provide an obstacle to designing peer prediction mechanisms based on this divergence based method.

In Sect.~\ref{sec:more}, we generalize Mechanism~\ref{alg:fmechansim} to more than two agents.  We show how machine learning techniques can be naturally integrated with our mechanism.

Finally, in Appendix~\ref{sec:comparison} we compare our mechanisms with \citet{Shnayder2016-xx} and \citet{kong2019information}.

\section{Preliminaries}\label{sec:pre}
We use $(\Omega,\mathcal{F}, \mu)$ to denote a measure space where $\mathcal{F}$ is a $\sigma$-algebra on the outcome space $\Omega$ and $\mu$ is a measure.  Let $\Delta_\Omega$  denote the set of distributions of over $(\Omega, \mathcal{F})$,\footnote{We assume these distribution has a density function with respect to the $\mu$, $P\ll\mu$ for all $P\in \Delta_\Omega$.   The distributions in $\Delta_\Omega$ depend on $\mathcal{F}$ and $\mu$, but we omit it to simplify the notation.  The density is defined as the Radon–Nikodym derivative $\frac{dP}{d\mu}$ which exists because $P$ is dominated by $\mu$.} and $\mathcal{P}$ as a subset of distributions in $\Delta_\Omega$.  Given a distribution $P$, we also use $P$ to denote the density function where $P(\omega)$ is the probability density of outcome $\omega\in \Omega$.  We use uppercase for a random object $X$ and lowercase for the outcome $x$.  
In this paper we consider $\Phi$ to be a convex continuous function and use $\dom(\Phi)$ to denote its domain.

\subsection{Mechanism Design for Information Elicitation}\label{sec:goals}
For simplicity we first consider two agents, Alice and Bob, who work on a set of $m$ tasks denoted as $[m]$.  For each task $s\in [m]$, Alice receives a signal $x_s$ in $\X$ and Bob a signal $y_s$ in $\Y$. We use $(\mathbf{X}, \mathbf{Y})\in (\X\times \Y)^{m}$ to denote the \emph{signal profile} of Alice and Bob which is generated from a prior distribution $\mathbb{P}$.\footnote{The prior can be subjective, and Alice's and Bob's can be difference.  Here we analyzes the process in Alice's perspective.}  In this paper, we make the following assumption:

\begin{assumption}[A priori similar tasks~\cite{dasgupta2013crowdsourced}]\label{ass:apriori} $\mathbb{P}$ is a prior, and each task is identically and independently (i.i.d.) generated: there exists a distribution $P_{X,Y}$ over $\X\times \Y$ such that $\mathbb{P} = P_{X,Y}^m$,
Moreover, we assume the marginal distributions have full supports, $P_X(x)> 0$ and $P_Y(y)> 0$ for all $x\in \X$ and $y\in \Y$.
\end{assumption}

Given a report profile of Alice, $\hat{\mathbf{x}}\in \X^m$ and Bob, $\hat{\mathbf{y}}\in \Y^m$, an \emph{information elicitation mechanism} $\mathcal{M} = (M_A, M_B)$ with $m$ tasks pays  $M_A(\hat{\mathbf{x}}, \hat{\mathbf{y}}) \in \R$ to Alice, and $M_B(\hat{\mathbf{x}}, \hat{\mathbf{y}}) \in R$ to Bob.  In the rest of the paper we often only define notions for Alice, and define Bob's in the symmetric way.

Besides Assumption~\ref{ass:apriori}, we assume their strategies are uniform and independent across different tasks which is also made in previous work~\cite{dasgupta2013crowdsourced,Shnayder2016-xx,kong2019information}.  Formally, the \emph{strategy} of Alice is a random function $\theta_A:\X \to \Delta_\X$ where $\theta_A(x,\hat{x})$ is the probability that Alice reports $\hat{x}$ conditioning on her private information $x$.  That is, each report only depends on the corresponding signal.  For instance, given Alice receiving $\mathbf{x}\in \X^m$ the probability that Alice reports $\hat{\mathbf{x}}\in \X^m$ is $\Pr[\hat{\mathbf{X}} = \hat{\mathbf{x}}] = \prod_{s\in [m]} \theta_A({x}_{s},\hat{x}_{s})$.
  We  call $\bm{\theta} = (\theta_A,\theta_B)$ a the \emph{strategy profile}.  The \emph{ex-ante payment} to Alice under a strategy profile $\bm{\theta}$ and a prior $\mathbb{P}$ in mechanism $\mathcal{M}$ is
$$u_A(\bm\theta;\mathbb{P}, \mathcal{M})\triangleq
   \E_{(\mathbf{X}, \mathbf{Y})}\left[\E_{(\hat{\mathbf{X}},\hat{\mathbf{Y}})}\left[\E_\mathcal{M}[M_A(\hat{\mathbf{x}}, \hat{\mathbf{y}})]\right]\mid (\mathbf{x}, \mathbf{y}) \right]$$ 
where we use a semicolon to separate the variable, $\theta$, and parameters $\mathbb{P}$ and $\mathcal{M}$.  Note that a strategy profile $\bm{\theta}$ can be seen as a Markov operator on probability measures on the signal space $\X\times\Y$, and Alice and Bob's reports, $\bm{\theta}\circ P$, is also a distribution on the signal space $\X\times\Y$.

In the literature of information elicitation, there are three important classes of strategies.  We use $\bm{\tau}$ to denote the  \defn{truth-telling strategy profile} where both agents' reports are equal to their private signals with probability $1$, e.g., Alice's strategy is $\tau_A(x,\hat{x}) = \mathbb{I}[x = \hat{x}]$.   A strategy profile is a \defn{permutation strategy profile} if both agents' strategy are a (deterministic) permutation, a bijection between signals and reports.  Finally, a strategy profile is \defn{oblivious} or \emph{uninformed} if even one of the agents' strategies does not depend on their signal: that is for Alice $\theta_A(x,\hat{x})=\theta_A(x',\hat{x})$ for all $x$, $x'$, and $\hat{x}$ in $\X$.  
Note that the set of permutation strategy profiles includes the truth-telling strategy profile $\bm{\tau}$ but does not include any oblivious strategy profiles. 
\paragraph{Truthful Guarantees}
We now define some truthfulness guarantees for our mechanism $\mathcal{M}$ that differ in how unique the high payoff of truth-telling strategy profile is:
\begin{description}
\item[Truthful:] the truth-telling strategy profile $\bm{\tau}$ is a Bayesian Nash Equilibrium, and has the highest payment to both Alice and Bob.
\item[Informed-truthful~\cite{Shnayder2016-xx}:]  Truthful and also for each agent $\bm{\tau}$ is strictly better than any oblivious strategy profiles.  For any oblivious strategy profile $\bm{\theta}$, $u_A(\bm{\tau}; \mathbb{P}, \mathcal{M})> u_A(\bm{\theta}; \mathbb{P}, \mathcal{M})$ and $u_B(\bm{\tau}; \mathbb{P}, \mathcal{M})> u_B(\bm{\theta}; \mathbb{P}, \mathcal{M})$.
\item[Strongly truthful~\cite{Shnayder2016-xx,kong2019information}:] Truthful and also for each agent $\bm{\tau}$ is strictly better  than all non-permutation strategy profiles.  For any non-permutation strategy profile $\bm{\theta}$, $u_A(\bm{\tau}; \mathbb{P}, \mathcal{M})> u_A(\bm{\theta}; \mathbb{P}, \mathcal{M})$ and $u_B(\bm{\tau}; \mathbb{P}, \mathcal{M})> u_B(\bm{\theta}; \mathbb{P}, \mathcal{M})$.
\item[Dominant truthful:] Each agent report truthfully leads to higher expected payoff than other strategies, regardless of other agent's reporting strategies.  For any strategy profile $\bm{\theta}$, we have $u_A(\bm{\tau}; \mathbb{P}, \mathcal{M})> u_A(\bm{\theta}; \mathbb{P}, \mathcal{M})$ and $u_B(\bm{\tau}; \mathbb{P}, \mathcal{M})> u_B(\bm{\theta}; \mathbb{P}, \mathcal{M})$.
\end{description}  
We can also call a general mapping truthful, informed-truthful, strongly truthful, dominant truthful when it satisfy the corresponding property.

In this work, we consider an approximate version of above statements with low sample number.  For example, given $\epsilon>0$, a mechanism $\mathcal{M}$ with $m(\epsilon)$ tasks (the sample number)\footnote{Here mechanism which can take different length of report $m$.  Or we can consider a family of mechanisms ($\mathcal{M}_m$) parameterized by the sample number (the number of tasks) $m$.} is \emph{$\epsilon$-strongly truthful with $m(\epsilon)$ tasks} if there exists a mapping from strategy profiles to ex-ante payments such that 1) this mapping is strongly truthful; 2) for all $\epsilon$ the ex-ante payments of our mechanism with $m(\epsilon)$ tasks is within $\epsilon$ of this mapping.

Now we define the sample number for approximately truthfulness guarantees. 
\begin{definition}\label{def:approx_strongly}
Given a family of joint signal distributions $\mathcal{P}$ and a function $S:\mathbb{R}_{>0}\to \mathbb{N}$ we say a mechanism $\mathcal{M}$ is \defn{$\epsilon$-strongly truthful on $\mathcal{P}$ with $S(\epsilon)$ number of tasks}, if there exists a strongly truthful mapping $F = (F_A, F_B)$ from joint signal distributions and strategy profiles to payments such that for all $\epsilon>0$ and $m\ge S(\epsilon)$
\begin{enumerate}
    \item the ex-ante payment under the truth-telling strategy profile in $\mathcal{M}$ with $m$ number of tasks is within $\epsilon$ additive error from $F$:
    for all $P\in \mathcal{P}$,
    $$u_A(\bm{\tau};P, \mathcal{M})\ge F_A(\bm{\tau},P)- \epsilon;$$
    \item and the ex-ante payment under any strategy profile $\bm{\theta}$ in $\mathcal{M}$ with $m$ number of tasks is bounded above by $F$: 
    for all $P\in \mathcal{P}$, and $\bm{\theta}$
    $$u_A(\bm{\theta};P, \mathcal{M})\le F_A(\bm{\theta},P).$$
\end{enumerate}
And the inequality also holds for Bob's ex-ante payment.  Furthermore, we say $\mathcal{M}$ is \defn{$(\delta, \epsilon)$-strongly truthful on $\mathcal{P}$ with $S(\delta, \epsilon)$} if the above conditions holds with probability $1-\delta$ for all $\delta\in (0,1)$ and $\epsilon>0$.

  Additionally, we say $\mathcal{M}$ is \defn{$\epsilon$-informal-truthful} (\defn{$\epsilon$-truthful}) with $S(\epsilon)$ number of tasks if it is $\epsilon$ close to an inform-truthful (truthful) mapping.

\end{definition}
Note that our notion of $\epsilon$-truthfulness guarantee is quite strong.  In particular, the second item requires for any strategy profile $\bm{\theta}$, the ex-ante payment is upper bounded by a strongly truthful (inform-truthful, truthful) mapping.

\subsection{Prior Assumptions}\label{sec:prior} 
There are two axes to compare these peer prediction mechanism: \emph{truthful guarantee} and \emph{prior assumption}. Truthful guarantee asks how good the truth-telling strategy is.  Prior assumption addresses how general these mechanisms are.   We first introduce the weakest possible notion of interdependence that we used in our paper.  Then we survey other notions proposed in previous works.  Finally, we provide concrete examples to show the distinction between those notions of interdependence.

\begin{definition}[Stochastic Relevant~\cite{Shnayder2016-xx}]\label{def:nondegenerate}
We call $P_{X,Y}$ \emph{stochastic relevant} if for any two distinct signals $x, x'\in \X$
$$P_{X,Y}[Y\mid X = x]\neq P_{X,Y}[Y\mid X = x'].$$
That is, Alice's posteriors on Bob's signals are different when Alice receives signal $x$ or $x'$.  And symmetrically, the same holds for Bob's  posterior on Alice's signals.
\end{definition}

Stochastic relevancy is the weakest assumption we can hope for designing peer prediction mechanisms. Proposition~\ref{prop:impossibility} shows that if agent's signal are not stochastic relevant an agent can always misreport regardless other agents' reports even if the mechanism knows the information structure.
\begin{prop}[Elicitability~\cite{zhang2014elicitability}]\label{prop:impossibility}
If the prior $P_{X,Y}$ is not stochastic relevant, there is no mechanism that has truth-telling as a strict Bayesian Nash equilibrium.
\end{prop}
Besides the above notion, previous peer prediction mechanisms make ad hoc assumptions on agents' private signals.  

\citet{kong2019information} studies \emph{fine-grained} joint distributions.  A joint distribution $P_{X,Y}$ is \emph{fine-grained} if for any distinct pairs of signals $(x,y)$ and $(x',y')$
$$\frac{P_{X,Y}(x,y)}{P_X(x)P_Y(y)}\neq \frac{P_{X,Y}(x',y')}{P_X(x')P_Y(y')}.$$
\citet{kong2020dominantly} considers \emph{strictly correlated} distributions.  A joint distribution $P$ on a finite space $\X^2$ is strictly correlated if the determinant of distribution $P\in \R^{|\X|\times|\X|}$ is nonzero.   Those two notions are both stronger than stochastic relevance, and the following example to show this.

\begin{example}
Suppose Alice and Bob review papers and grade with a scale from one to three (reject, neutral, accept) $\{1,2,3\}$.  The paper can be either good or bad with uniform probability.  If the paper is good, Alice's grade $X\in \{1,2,3\}$ is generate from distribution $\Gamma_1 = [0.2, 0.2, 0.6]^\top$.  If the paper is bad, $X$ is sampled from  distribution $\Gamma_0 = [0.6, 0.2, 0.2]^\top$.  The distribution of Bob's signal $Y$ on the paper is sampled identically and independently conditional on the state of the paper.  Thus, $P_{X,Y}(x,y) = 0.5\Gamma_{1,x}\Gamma_{1,y}+0.5\Gamma_{0,x}\Gamma_{0,y}$, and the joint distribution of $X$ and $Y$ is 
\begin{equation}\label{eq:counter}
    P_{X,Y} = 0.5\Gamma_1^\top \Gamma_1+0.5 \Gamma_0^\top \Gamma_0 = \begin{pmatrix}
    0.20&    0.08&    0.12\\
    0.08&    0.04&    0.08\\
    0.12&    0.08&    0.20
\end{pmatrix}
\end{equation}
where the first index is on $X$ and the second index is on $Y$.

Now we show this simple prior is stochastic relevant but not fine-grained nor strictly correlated.  Therefore those mechanisms do not have provable truthful guarantee on this simple prior.
\begin{enumerate}
    \item The prior $P_{X,Y}$ in Eqn.\eqref{eq:counter} is stochastic relevant, because $\Pr[Y\mid X] = \begin{pmatrix}
        0.5&    0.2&    0.3\\
    0.4&    0.2&    0.4\\
    0.3&    0.2&    0.5
    \end{pmatrix}$ where each row is distinct and $\Pr[Y\mid X]$ also has this property  due to symmetry.  
    \item $P$ is not fine-grain, because $\left(\frac{P_{X,Y}(x,y)}{P_X(x) P_Y(y)}\right)_{x,y\in \X\times \Y} = \begin{pmatrix}
    1.25&    1.00&    0.75\\
    1.00&    1.00&    1.00\\
    0.75&    1.00&    1.25
    \end{pmatrix}$ where $\frac{P_{X,Y}(1,0)}{P_X(1)P_Y(0)} =  \frac{P_{X,Y}(0,1)}{P_X(0)P_Y(1)}$.
    \item Finally, $P_{X,Y}$ is not strictly correlated, because $det(P_{X,Y}) = 0$.
\end{enumerate}
\end{example}

\subsection{Convex Analysis and \texorpdfstring{$\Phi$}{Phi}-divergence}\label{sec:pre_convex}
Informally, $\Phi$-divergences quantify the difference between a pair of distributions over a common measurable space.  
\begin{definition}[$\Phi$-divergence~\cite{csiszar1964informationstheoretische, Morimoto1963, ali1966general}]\label{def:fdiv} Let $\Phi:[0,\infty)\to \R$ be a convex function with $\Phi(1) = 0$.  Let $P$ and $Q$ be two probability distributions on a common measurable space $(\Omega, \mathcal{F})$.  The \defn{$\Phi$-divergence of $Q$ from $P$} where $P\ll Q$\footnote{$P$ is absolutely continuous with respect to $Q$: for any measurable set $A\in \mathcal{F}$, $Q(A) = 0\Rightarrow P(A) = 0$.} is defined as
$D_\Phi(P\| Q)\triangleq \E_Q\left[\Phi\left(P/Q\right)\right].$\footnote{$P/Q$ is the Radon-Nikodym derivative between measures $P$ and $Q$, and it is equal to the ratio of density function.}
\end{definition}

We can use these divergences to measure how interdependent between two random variables $X$ and $Y$.  Formally, Let $P_{X,Y}$ be a distribution over $(x,y)\in \X\times\Y$, and $P_X$ and $P_Y$ be marginal distributions of $X$ and $Y$ respectively.  We set $P_X P_Y$ be the tensor product between $P_X$ and $P_Y$ such that $P_X P_Y(x,y) = P_X(x)P_Y(y)$.  We call $D_\Phi(P_{X,Y}\|P_X P_Y)$ the \defn{$\Phi$-mutual information between $X$ and $Y$}.  

Given a joint distribution $P_{X,Y}$, let \defn{\Ratio} at $(x,y)$ on  $P_{X,Y}$ be $$\ratio_P(x,y):=\frac{P_{X,Y}(x,y)}{P_X(x)P_Y(y)}$$
which is ratio between joint probability divided by the product of the probabilities at $(x,y)$.  We will omit subscript $P$ when there is not ambiguity.   This ratio appears in several literature.  For instance, it's called  \textit{observed to expected ratio} in life sciences literature, or \textit{lift} in data mining for binary random variable.  Additionally, $\log \ratio(x,y)$ is called point-wise mutual information.  Finally, note that $\Phi$ mutual information is the average of \Ratio applied to $\Phi$.

Now, we introduce some basic notions in convex analysis~\cite{rockafellar2015convex}.
Let $\Phi:[0,+\infty)\to \R$ be a convex function.  The \emph{convex conjugate} $\Phi^*$ of $\Phi$ is defined as:
$\Phi^*(b) = \sup_{a\in \dom(\Phi)}\{ab-\Phi(a)\}$.
Moreover $\Phi = \Phi^{**}$ if $\Phi$ is continuous.

By Young-Fenchel inequality~\cite{fenchel1949conjugate}, we can rewrite the $\Phi$-divergence of $Q$ from $P$ in a variational form.  This formulation is important to understand our mechanisms.

\begin{theorem}[Variational representation~\cite{nguyen2010estimating, wu2017lecture}]\label{thm:fdivergenceconjugate}
$$D_\Phi(P\|Q) = \sup_{k:\Omega\to \dom(\Phi^*)} \left\{\E_{\omega\sim P}[k(\omega)]-\E_{\omega\sim Q}[\Phi^*(k(\omega))]\right\},\footnote{The $\sup$ is taken over $k$ with finite $\E_{\omega\sim P}[k(\omega)]$ and $\E_{\omega\sim Q}[\Phi^*(k(\omega))]$.}$$
and the equality holds $D_\Phi(P\|Q) = \E_{\omega\sim P}[k(\omega)]-\E_{\omega\sim Q}[\Phi^*(k(\omega))]$ if and only if $k \in \partial \Phi\left(P/Q\right)$ almost everywhere on $Q$.\footnote{$\partial \Phi$ is the subgradient of $\Phi$, and the formal definition can be found in \cite{rockafellar2015convex}.  Here we only use the equality condition when $\Omega$ is finite.}
\end{theorem}

For completeness, we provide a proof for Theorem~\ref{thm:fdivergenceconjugate} and some examples for $\Phi$-divergence in Appendix~\ref{sec:convex_add}.

\subsection{Scoring Function}
Our constructions and analysis will make heavy use of the following functionals--- scoring functions.
\begin{definition}[Scoring function]\label{def:score}
A \defn{scoring function} $K:\X\times\Y\to \R$ is a functional (real-valued function) that maps from a pair of reports to a real value.  Given a convex function $\Phi$, a scoring function $K^\star_{P, \Phi}$ is a \defn{$(P_{X,Y},\Phi)$-ideal scoring function} if
\begin{equation}\label{eq:optimal}
    K^\star_{P, \Phi} (x,y) \in \partial \Phi\left(\frac{P_{X,Y}(x,y)}{P_X(x)P_Y(y)}\right) = \partial \Phi(\ratio_P(x,y)).
\end{equation}
We will use $P$ and $P_{X,Y}$ interchangeably later, and say $K^\star$ is \emph{ideal} without specifying $P$ and $\Phi$ when it's clear. 
\end{definition}
A $(P, \Phi)$-ideal scoring function is the \Ratio applied to $\partial \Phi$ which is  a monotone increasing function if $\Phi$ is differentiable.  \Ratio encodes the signal structure of $P_{X,Y}$ which measure how interdependent $x$ and $y$ is.  Alternatively, the scoring function serves as a ``distinguisher'' which tries to decide whether a pair of reports came from the joint distribution or the product of the marginal distributions. 

Furthermore, the ideal scoring function can be easily computed from the density function $P_{X,Y}$.  We give a example that will serve as a running example in this paper.

\begin{example}[Joint Gaussian Signals]
On each day $s$, a certain route has a expected driving time $\mu_s$ drawn from Gaussian distribution $\mathcal{N}(m_0, \sigma^2)$ i.i.d., \footnote{$\mathcal{N}(m_0, \sigma^2)$ denotes the Gaussian distribution with mean $m_0$ and covariance matrix (or variance) $\sigma^2$} and Alice receives a driving time $x$ from $\mathcal{N}(\mu_s, \tau^2)$ and Bob receives $y$ from $\mathcal{N}(\mu_s, \tau^2)$ independently conditioned on $\mu_s$.  Therefore, $P_{X,Y}$ is pair of correlated Gaussians with mean $(m_0, m_0)$ and covariance $\begin{pmatrix}
\sigma^2+\tau^2 & \sigma^2\\
\sigma^2    & \sigma^2+\tau^2
\end{pmatrix}$.  Let $G(x,y) \triangleq (x-m_0, y-m_0)\begin{pmatrix}{\sigma^2+\tau^2} & -\sigma^2\\-\sigma^2    & {\sigma^2+\tau^2}\end{pmatrix}\begin{pmatrix}x-m_0 \\ y-m_0
\end{pmatrix}$ be a quadratic form on $x$ and $y$.  Then the \Ratio  is
$$\ratio(x,y) = \frac{P_{X,Y}(x,y)}{P_X(x)P_Y(y)} = \sqrt{\frac{(\sigma^2+\tau^2)^2}{2\sigma^2\tau^2+\tau^4}}\exp\left(\frac{-1}{2(2\sigma^2\tau^2+\tau^4)}G(x,y)\right).$$

If $\Phi(a) = \frac{1}{2}|a-1|$, and constant $R \triangleq \left(\sigma^2\tau^2+\tau^4\right)\log\left(\frac{(\sigma^2+\tau^2)^2}{2\sigma^2\tau^2+\tau^4}\right)$, an ideal scoring function is
$$K^\star_{P, \Phi} (x,y) = \begin{cases}\frac{1}{2} \text{ if } G(x,y)< R \\
-\frac{1}{2} \text{ if } G(x,y)\ge R
\end{cases}$$
which can be represented by an ellipse $\Gamma$. The scoring function is $1/2$ if the input is in the ellipse and $-1/2$ otherwise. (cf. Figure~\ref{fig:example}) 

If $\Phi(a) = a\log a$, the $\Phi$-ideal scoring function is
$$K^\star_{P, \Phi} (x,y) = -\frac{1}{2(2\sigma^2\tau^2+\tau^4)}G(x,y)+1+\frac{1}{2}\log \left(\frac{(\sigma^2+\tau^2)^2}{2\sigma^2\tau^2+\tau^4}\right)$$
which is a quadratic function on $x$ and $y$. (cf. Figure~\ref{fig:example})

\begin{figure}[ht]
\centering
\begin{tabular}{ccc}
\includegraphics[width=0.32\textwidth]{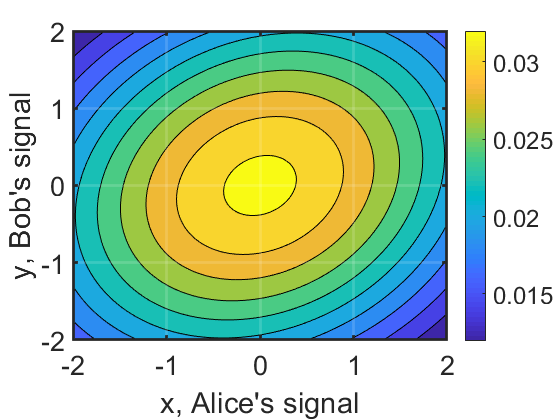} &\includegraphics[width=0.32\textwidth]{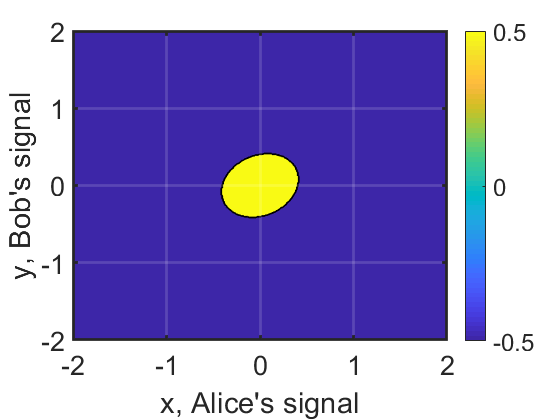}
&\includegraphics[width=0.32\textwidth]{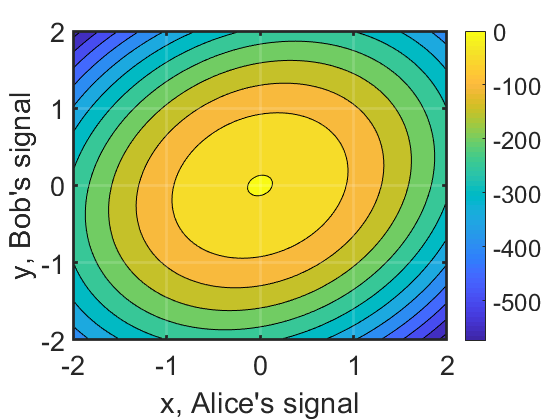}\\
\includegraphics[width=0.32\textwidth]{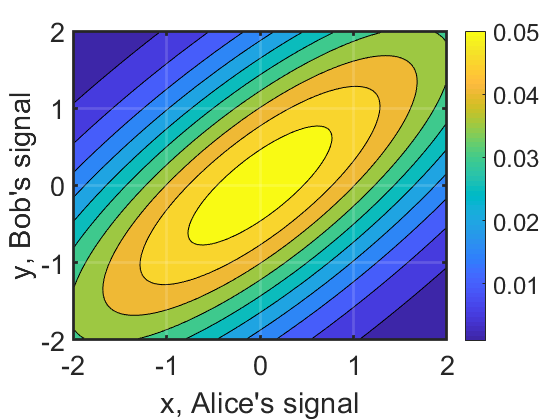} &\includegraphics[width=0.32\textwidth]{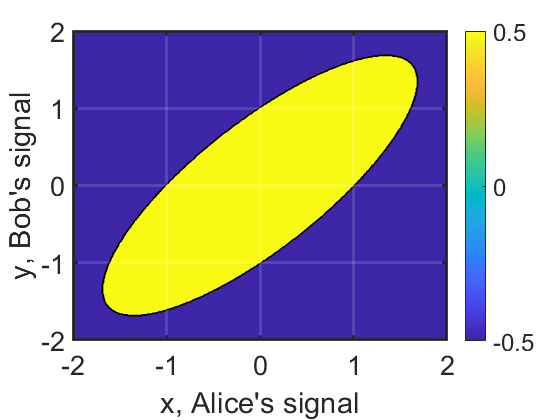}
&\includegraphics[width=0.32\textwidth]{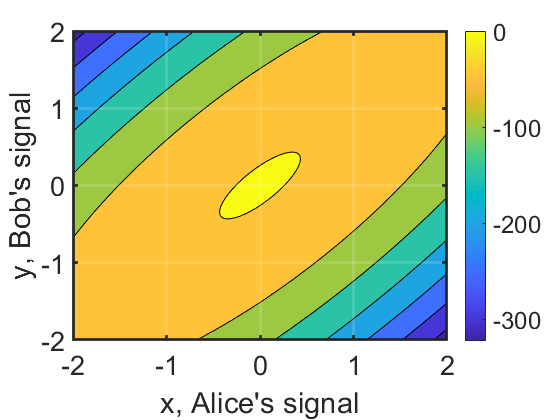}\\
$P_{X,Y}$ probability density function  & $K^\star_{P, \Phi}$ with $\Phi(a) = |a-1|$ & $K^\star_{P, \Phi}$ with $\Phi(a) = a\log a$\\
\end{tabular}
\caption{The top row uses $\sigma = 1$ and $\tau = 2$, and the bottom row uses $\sigma = 2$ and $\tau = 1$. Note that if Alice's and Bob's signals are more correlated $\sigma\gg\tau$, $\Gamma$ is more skew}\label{fig:example}
\end{figure}
\end{example}

\subsection{Functional Complexity}
In thi section, we provide some standard notions to characterize the complexity of learning functionals which are standard~ \cite{van2000empirical, wellner2013weak}, and discuss more in Sect.~\ref{sec:pre_emp}.  We will use these notions to characterize the complexity of learning an ideal scoring function.

Let $\mathcal{K}$ is a pre-specified class of functionals $k:\X\times \Y\to \R$.  Given $k\in \mathcal{K}$, $L>0$, and a distribution $P_{X,Y}$, we define the \emph{Bernstein norm} as $$\rho_L^2(k;P)\triangleq 2L^2\E_{P}[ \exp\left(|k|/L\right)-1-|k|/L]\text{,
and }\rho_L(\mathcal{K};P) \triangleq \sup_{k\in\mathcal{K}} \rho_{L}(k,P).$$
Let  $\mathcal{N}_{[],L}(\delta, \mathcal{K}, P)$ be the smallest value of $n$ for which there exists $n$ pairs of functions $\{(k_j^L, k_j^U)\}$ such that 1) $\rho_L(k_j^U-k_j^L;P)\le \delta$ for all $j$ and 2) for all $k\in \mathcal{K}$ there is a $j$,  $k_j^L(x,y)\le k(x,y)\le k_j^U(x,y)$ for all  $(x,y)\in\X\times\Y$.
Then 
$$\mathcal{H}_{[],L}(\delta, \mathcal{K}, P) \triangleq \log \mathcal{N}_{[],L}(\delta, \mathcal{K}, P)$$
is called the \emph{generalized entropy with bracketing}.  We further define the entropy integral as $J_{[], L}(R,\mathcal{K}, P) \triangleq \int^{R}_0\sqrt{\mathcal{H}_{[],L}(u,\mathcal{K},P)}du$.

Our results will show that constant number of questions suffice as long as the ideal scoring functions is in some bounded complexity space $\mathcal{K}$ where $J_{[], L}(R,\mathcal{K}, P)$ and $\rho_L(\mathcal{K};P)$ are bounded.

\section{\texorpdfstring{$\Phi$}{Phi}-Divergence Pairing Mechanisms}\label{sec:pairing}
In this section, we first define a class of multi-task peer-prediction mechanisms $\mathcal{M}^{\Phi,K}$  Alice and Bob who work on all $m\ge 2$ tasks.  The mechanism is parametrized by a convex function $\Phi$ and a scoring function $K$ (Definition~\ref{def:score}).  Then we briefly discuss how to obtain a good scoring function, and develop algorithms for estimating good scoring function.

The process of this mechanism is quite simple.  Given a scoring function $K$ and $\Phi$, we arbitrarily choose one task $b$, and two distinct tasks  $p$ and $q$ from $m\ge2$ tasks.  Alice gets paid by Eqn.~\eqref{eq:payment} the scoring function on her and Bob's reports on task $b$ minus the $\Phi^*$ applied to the scoring function on her report on $p$ and Bob's report on $q$.  In this way, agents are paid by a scoring function on a \emph{correlated task} minus a regularized scoring function on two \emph{uncorrelated tasks}.

\begin{algorithm}[hbt!]
\floatname{algorithm}{Mechanism}
  \caption{$\Phi$-divergence pairing mechanism with a scoring function $K$ for two agents, $\mathcal{M}^{\Phi,K}$}\label{alg:fmechansim}
  \begin{algorithmic}[1]
    \REQUIRE   A report profile  $(\hat{\mathbf{x}}, \hat{\mathbf{y}})$ where both Alice and Bob submit report for all $m\ge 2$ tasks.
    \ENSURE A convex function $\Phi:[0,\infty)\to \mathbb{R}$, its conjugate $\Phi^*$, and a  scoring function $K:\X\times\Y\to \dom(\Phi^*)\subseteq \mathbb{R}$.
    \STATE For Alice, arbitrarily pick three tasks $b$, $p$ and $q$ where $p$ and $q$ are distinct. {We call $b$ the \emph{bonus task}, $p$ the \emph{penalty task to Alice}, and $q$ the \emph{penalty task to Bob}.}
    \STATE Based on Alice's reports on $b$ and $p$ ($\hat{x}_b$ and $\hat{x}_{p}$) and Bob's reports on $b$ and $q$ ($\hat{y}_b$ and $\hat{y}_{q}$), the payment to Alice is
\begin{equation}\label{eq:payment}
M_A^{\Phi,K}(\hat{\mathbf{x}}, \hat{\mathbf{y}}) \triangleq K\left(\hat{x}_{b}, \hat{y}_{b}\right)-\Phi^*\left(K\left(\hat{x}_{p}, \hat{y}_{q}\right)\right).
\end{equation}
    \STATE The payment of Bob is defined similarly.
  \end{algorithmic}
\end{algorithm}

To simplify the notion, we use $u_A$ or $u_A(\bm{\theta},P,K)$ to denote the ex-ante payment to Alice under a strategy profile $\bm{\theta}$ and a joint signal distribution $P$ in pairing mechanism with a scoring function $K$.

In general, the truthfulness guarantees of Mechanism~\ref{alg:fmechansim} depends on the degeneracy of Alice's and Bob's signal distribution $P$ and convex function $\Phi$.  In this paper, we consider three different conditions which will be used in the statement of our results.
\begin{assumption}\label{ass:degenerate}
In this paper, we consider the following four different settings.
\begin{enumerate}
    \item no assumption;
    \item $P_{X,Y}$ is stochastic relevant;\label{item:degenerate2}
    \item Besides the above conditions, $\X$ and $\Y$ are finite sets, $\Phi$ is strictly convex and differentiable, and $\Phi^*$ is strictly convex.\label{item:degenerate3}
\end{enumerate}
\end{assumption}
\subsection{Obtaining a Good Scoring Function}
The $\Phi$-pairing mechanism $\mathcal{M}^{\Phi,K}$ is not stand-alone mechanism for information elicitation, because it requires a scoring function $K$ as a parameter.  We will see shortly in Sect.~\ref{sec:known_prior} and \ref{sec:free}, the truthfulness guarantees of the pairing mechanism depends on the quality of the scoring function.  In this paper, we consider three different models for mechanism designers to estimate good scoring functions which are discussed in the rest of the sections:

\begin{description}
\item[Direct access of $K^\star_{P, \Phi}$]  In Sect.~\ref{sec:known_prior}, we first consider the mechanism knows a $(P,\Phi)$-ideal scoring function $K^\star_{P, \Phi}$.  Note that if the mechanism knows the prior $P$, it can compute the $(P,\Phi)$-ideal scoring function, but the converse is not necessarily true.
\item[General reduction to a learning problem]  In Sect.~\ref{sec:free}, besides the reports from Alice and Bob, mechanism may exploit Alice and Bob's previous scoring function and other side information.  For example the joint distribution between Alice and Bob can be approximated by some parametric model, say joint Gaussian distributions.  We introduce our framework (Mechanism~\ref{alg:reduction}) that reduces the problem into a learning problem.
\item[Estimation from samples]  Finally, in the multi-task setting, if Alice and Bob truthfully report their signals, it is possible to estimate the $(P,\Phi)$-ideal scoring function from those reports.  However, the mechanism needs to incentive them to be truthful.   In Sect.~\ref{sec:learning}, we propose two learning methods to estimate good scoring functions.  Combining them with our framework (Mehcanism~\ref{alg:reduction}), we can have detail-free $\epsilon$-strongly truthful mechanisms with high probability.
\end{description}

\section{Pairing Mechanisms in the Known Prior Setting}\label{sec:known_prior}
If the the mechanism $\mathcal{M}^{\Phi, K^\star}$ has an $(P,\Phi)$-ideal scoring function $K^\star$ where $P$ is the joint distribution to Alice's and Bob's signals, the mechanism has the following properties.  We defer the proof to Sect.~\ref{sec:lem}.
\begin{theorem}\label{thm:truth}
Let an integer $m$ be greater than $2$, a functional $\Phi$ be a continuous convex function with $[0,\infty)\subseteq \dom(\Phi)$, $\mathbb{P}$ with $P_{X,Y}$ be a common prior between Alice and Bob satisfying Assumption~\ref{ass:apriori}.  Let  $\bm{\tau}$ be the truth-telling strategy profile, and  $K^\star$ be a $(P,\Phi)$-ideal scoring function.

The $\Phi$-pairing mechanism with $K^\star$, $\mathcal{M}^{\Phi, K^\star}$ has the following properties: For any strategy profile $\bm{\theta}$, \footnote{There are some minor details when $\X$ and $\Y$ are not finite set.  Here we require $\bm{\theta}$ to have finite $\int H\,d\theta_A \,d\theta_B dP_{X,Y}$,and $\int \Phi^*(H)\,d\theta_A\,d\theta_B dP_X P_Y$.} \begin{equation}\label{eq:truth1}
    u_A\left(\bm{\theta},P,H\right)\le u_A\left(\bm{\tau},P,H\right).
\end{equation}
Furthermore, under the four conditions in Assumption~\ref{ass:degenerate} respectively, the mechanism $\mathcal{M}^{\Phi, K^\star}$ is 
\begin{enumerate}
    \item truthful,
    \item informed-truthful, or
    \item strongly truthful.
\end{enumerate}
\end{theorem}

In the following example, we show how Mechanism~\ref{alg:fmechansim} with a $(P,\Phi)$-ideal scoring function works, and illustrate the difference between informed-truthful and strongly truthful.

\begin{example}[continued]
On each day $s$, Alice and Bob learn their commute time $(x_s, y_s)\in \R^{2}$.  We want to use Mechanism~\ref{alg:fmechansim} to collect those commute time, and we know $P_{X,Y}$. 

When $\Phi(a) = \frac{1}{2}|a-1|$, a $(P,\Phi)$-ideal scoring function is $K^\star(x,y) = \mathbb{I}[G(x,y)>R]-1/2$ which can be represented by an ellipse $\Gamma$.
After Alice and Bob report their every day's commute time $\hat{\mathbf{x}},\hat{\mathbf{y}}$, the mechanism arbitrarily pick a bonus day $b$, and two distinct penalty days $p$ and $q$.  Then it pays Alice with $1$ if their bonus day reports are in the ellipse $\Gamma$ and their penalty days reports are not in $\Gamma$.  As seen in Fig.~\ref{fig:example}, $\Gamma$ is skew in diagonal, so if Alice's and Bob's reports on the bonus day are more correlated they can get more payment.  

However, if Alice receives an extremely large value (e.g. $x_s= 20$) such that she knows the scoring function $K^\star$ is $-1/2$ for certain regardless of Bob's report (cf. Figure~\ref{fig:example}), Alice can misreport her signal (e.g. $\hat{x}_s = 2$ when $x_s\ge 20$) without changing her expected utility.  Therefore the $\Phi$-pairing mechanism with $\Phi(a) = \frac{1}{2}|a-1|$ is not strongly truthful.  Additionally, truth-telling is not even a strict Bayesian Nash equilibrium.

To prevent Alice from truncating signals, instead of $\Phi(a) = \frac{1}{2}|a-1|$ we can take other strictly convex $\Phi$.  For example if $\Phi(a) = a\log a$, the ideal scoring function is a quadratic function and above-mentioned strategy cannot trivially hold.  In Theorem~\ref{thm:truth} we prove this in the finite signal spaces setting. 
\end{example}
\begin{remark}
Although the $\Phi$-pairing mechanism with a $(P,\Phi)$-ideal scoring function has many desirable properties shown in Theorem~\ref{thm:truth}, such a mechanism is not detail-free.  Furthermore, in the detail-free setting where mechanisms only access Alice's and Bob's reports, it is impossible to have a mechanism which has truth-telling strategy profile as the uniquely best equilibrium.  Informally, in the detail-free setting a mechanism $\mathcal{M}$ cannot distinguish between the following two situations: 
1) Alice and Bob's signals joint distribution is $P$ and their strategy profile is a permutation $\bm{\theta}$;
2)  Alice and Bob's signals joint distribution is $\bm{\theta}\circ P$ and their strategy profile is the truth-telling strategy, because their reports are generated from the same distribution $\bm{\theta}\circ P$ in both cases.  Therefore,
\begin{equation}\label{eq:associative}
    u_A(\bm{\theta}; P, \mathcal{M}) = u_A(\bm{\tau};\bm{\theta}\circ P, \mathcal{M}).
\end{equation}
Suppose  the ex-ante payment under the truth-telling strategy profile and $P$ is strictly higher than the ex-ante payment under a permutation strategy profile $\bm{\theta}$.  Then the ex-ante payment under a permutation strategy profile $\bm{\theta}^{-1}$ and the joint signal distribution $\bm{\theta}\circ P$, $u_A(\bm{\theta}^{-1};\bm{\theta}\circ P, \mathcal{M}) = u_A(\bm{\tau};P, \mathcal{M})$ is strictly higher than the ex-ante payment under truth-telling strategy profile  $u_A(\bm{\tau};\bm{\theta}\circ P, \mathcal{M}) = u_A(\bm{\theta};P, \mathcal{M})$.
This argument is trivially true when $\mathcal{M}$ is a $\Phi$-pairing algorithm and the  scoring function is a function of Alice's and Bob's reports.  For general detail-free mechanisms, the reader may refer to Sect.~8 of \citet{kong2019information}.  
\end{remark}
\section{Main Technical Lemmas}\label{sec:lem}
To prove Theorem~\ref{thm:truth}, we use the following lemmas which are also important in the rest of the paper.

We first show the ex-ante payment under the truth-telling strategy profile in the $\Phi$-pairing mechanism with $(P,\Phi)$-ideal scoring function is the $\Phi$-mutual information between Alice's and Bob's signals.
\begin{restatable}[Truth-telling]{lem}{truthtelling}\label{lem:truthtelling}
If $K^\star$ is a ($P_{X,Y}, \Phi$)-ideal scoring function,
$$u_{A}(\bm{\tau},P,K^\star) = D_\Phi(P_{X,Y}\| P_X P_Y).$$
\end{restatable}
Moreover, if $P_{X,Y}$ is stochastic relevant, 
$D_\Phi(P_{X,Y}\| P_X P_Y)>0$.

Then we show any deviation from the truth-telling strategy profile or an ideal scoring function cannot improve Alice (and Bob's) ex-ante payment.  The proof uses the variational representation of $\Phi$-divergence (Theorem~\ref{thm:fdivergenceconjugate}).
\begin{restatable}[Manipulation in strategies and scoring functions]{lem}{manipulation}
\label{lem:manipulation}
For any strategy profile $\bm{\theta}$ and scoring function $K$,\footnote{There are some minor details when $\X$ and $\Y$ are not finite set.  Here we require $K$ and $\bm{\theta}$ to have finite $\int K\,dP_{X,Y}$, $\int \Phi^*(K)d(P_{X}P_Y)$, $\int K\,d\theta_A \,d\theta_B dP_{X,Y}$ ,and $\int \Phi^*(K)\,d\theta_A\,d\theta_B dP_X P_Y$.} 
$$u_{A}(\bm{\theta},P,K) \le D_\Phi(P_{X,Y}\| P_X P_Y).$$
\end{restatable}

Note that combining these two lemmas we have an even stronger result than inequality~\eqref{eq:truth1} which is a key tool in this paper: 
For any scoring function $K$ and strategy profile $\bm{\theta}$, 
\begin{equation}\label{eq:upperbound}
     u_A\left(\bm{\theta},\mathbb{P},K\right)\le u_A\left(\bm{\tau},\mathbb{P},K^\star\right).
\end{equation}

\begin{restatable}[Oblivious strategy]{lem}{oblivious}\label{lem:oblivious}  If $\bm{\theta}$ is an oblivious strategy profile, for any scoring function $K$
$$u_{A}(\bm{\theta}, P, K) \le 0.$$
\end{restatable}

\begin{restatable}{lem}{eqauality}\label{lem:strict}
Moreover, given Conditions~\ref{item:degenerate3} in Assumption~\ref{ass:degenerate}, the equality in \eqref{eq:upperbound} for Alice or Bob occurs if and only if
\begin{enumerate}
    \item $\bm{\theta} = (\pi_A, \pi_B)$ which is a permutation strategy profile, and
    \item For all $x\in \X$ and $y\in \Y$, $K(\pi_A(x),\pi_B(y)) = \Phi'\left(\ratio(x,y)\right) = \Phi'\left(\frac{P_{X,Y}(x,y)}{P_{X}(x)P_{Y}(y)}\right)$.
\end{enumerate}
\end{restatable}
Informally, Lemma~\ref{lem:strict} shows if the pair of a strategy profile and a scoring function $(\bm{\theta},K)$ have \eqref{eq:upperbound} equal only if there is a ``conjugated'' structure between the strategy and the scoring function.  The proof uses the pigeonhole principle on the finite signal spaces and shows if the equality holds under a non permutation strategy profile, $P$ is not stochastic relevant.

\subsection*{Proof of Theorem~\ref{thm:truth}}
With the above four lemmas, we are ready to prove Theorem~\ref{thm:truth}.
\begin{proof}[Proof of Theorem~\ref{thm:truth}] 
There are four statements to show.

First, \eqref{eq:truth1} is a direct result of \eqref{eq:upperbound}.  Furthermore, \eqref{eq:truth1} proves that truth-telling is a Bayesian Nash equilibrium, and has highest ex-ante payment to Alice.\footnote{Note that without additional assumption the truth-telling is not a strict Bayesian Nash equilibrium.  This is illustrated in the example in Sect.~\ref{sec:known_prior}.}  This shows the mechanism is truthful.

By Lemma~\ref{lem:oblivious}, the ex-ante payment to Alice (and Bob) is non-positive.  Combining this and Lemma~\ref{lem:truthtelling}, we prove the $\Phi$-pairing mechanism with ($P, \Phi$)-ideal scoring function is inform-truthful when $P$ is stochastic relevant.

To show our mechanism is strongly truthful, under Condition~\ref{item:degenerate3} in Assumption~\ref{ass:degenerate}, we use the first part of Lemma~\ref{lem:strict}. If the ex-ante payment under some strategy profile is equal to the ex-ante payment under the truth-telling strategy profile, the strategy profile is a permutation strategy profile.  
\end{proof}

\section{The Pairing Mechanism in the Detail Free Settings}\label{sec:free}  
With Sect.~\ref{sec:lem}, we can see that to achieve the truthfulness guarantees, it suffices to have a ``good'' scoring function.  That is if the ex-ante payment to Alice under the truth-telling strategy profile is close to the $\Phi$-mutual information between Alice's and Bob's signals, by \eqref{eq:upperbound}, the ex-ante payment under an untruthful-strategy is less than the ex-ante payment under the  truth-telling strategy profile.

In Sect.~\ref{sec:approx} we formalize the notions of a \emph{good} scoring function and of the  \emph{accuracy} of a learning algorithm $\mathcal{L}$ for scoring functions.  In Sect.~\ref{sec:framework}, we state our main result, Theorem~\ref{thm:framework}, which reduces the mechanism design problem to a learning problem for an ideal scoring function, and provides some intuition about the proof of the theorem.

\subsection{Accuracy of Scoring Rules and Learning Algorithms}\label{sec:approx}
Now we define a \emph{good} scoring function, and the \emph{accuracy} of a learning algorithm $\mathcal{L}$.
Given $\Phi$, a prior $P_{X,Y}$ and $\epsilon>0$, we say that a scoring function $K$ is \defn{$\epsilon$-ideal on ($P_{X,Y}, \Phi$)}, if for Alice
\begin{equation}\label{eq:approximated}
    u_{A}(\bm{\tau},P, K) \ge u_{A}(\bm{\tau},P, K^\star_{P, \Phi})-\epsilon = D_\Phi(P_{X,Y}\| P_X P_Y)-\epsilon,
\end{equation}
and the similar inequality holds for Bob.  
Additionally, For $m_L\in \mathbb{N}$, we say a \emph{learning algorithm for scoring functions with $m_L$ samples}, as a function from $(\mathbf{x}_L, \mathbf{y}_L)\in (\X\times\Y)^{m_L}$ to a scoring function $K$.  Given $\mathcal{P}$, a set of distributions on $\X\times\Y$, and a function $S_L:\mathbb{R}\times\mathbb{R}\to \mathbb{N}$, we say such a learning algorithm $\mathcal{L}$ is \defn{$(\delta, \epsilon)$-accurate on $(\mathcal{P}, \Phi)$ with $S_L(\delta, \epsilon)$ samples}, if for all $P_{X,Y}\in \mathcal{P}$, $\delta\in (0,1)$, $\epsilon>0$, and $m_L\ge S_L(\delta, \epsilon)$:
$$\Pr_{(\mathbf{x}_L, \mathbf{y}_L)\sim P_{X,Y}^{m_L}}\left[ u_{A}(\bm{\tau}, P, \mathcal{L}(\mathbf{x}_L, \mathbf{y}_L)) > D_\Phi(P_{X,Y}\| P_X P_Y)-\epsilon\right]\ge 1-\delta.$$
That is given $m_L$ i.i.d. samples from $P_{X,Y}$, the probability that the output, $\mathcal{L}(\mathbf{x}_L, \mathbf{y}_L)$, is $\epsilon$-ideal on $(P,\Phi)$ is greater than $1-\delta$.  Note that we require the algorithm $\mathcal{L}$ approximates the ideal scoring \emph{uniformly} on all distributions in $\mathcal{P}$.

\subsection{Pairing Mechanism with Learning Algorithms}\label{sec:framework}
Now we replace a fixed scoring function with an accurate learning algorithm $\mathcal{L}$ in Mechanism~\ref{alg:fmechansim}.  Intuitively, in the detail-free setting, the Mechanism~\ref{alg:reduction} first runs a learning algorithm on Alice's and Bob's report profile to derive a scoring function, and then pays Alice and Bob by Mechanism~\ref{alg:fmechansim}. 
\begin{algorithm}[hbt!]
\floatname{algorithm}{Mechanism}
  \caption{$\Phi$-divergence pairing mechanism with a learning algorithm $\mathcal{M}^{\Phi, \mathcal{L}}$}\label{alg:reduction}
  \begin{algorithmic}[1]
    \ENSURE A convex function $\Phi$, and a learning algorithm $\mathcal{L}$ with $m_L$ samples.    
    \REQUIRE   A report profile $(\hat{\mathbf{x}}, \hat{\mathbf{y}})$ from Alice and Bob on $m$ tasks where $m\ge 2+m_L$.

    \STATE Partition $m$ tasks (arbitrarily) into a set of learning tasks $M_L$ and a set of scoring tasks $M_S$ where $|M_L|\ge m_L$ and $|M_S|\ge 2$.  Let $(\hat{\mathbf{x}}_L,\hat{\mathbf{y}}_L)$ be the reports from Alice and Bob on the learning tasks $M_L$, and $(\hat{\mathbf{x}}_S,\hat{\mathbf{y}}_S)$ be the reports on the scoring tasks.
    \STATE Run the learning algorithm and derive $K_{\rm est} = \mathcal{L}(\hat{\mathbf{x}}_L,\hat{\mathbf{y}}_L)$.
    \STATE Run the $\Phi$-pairing mechanism (Mechanism~\ref{alg:fmechansim}) with the scoring function $K_{\rm est}$, and pay Alice and Bob accordingly.
  \end{algorithmic}
\end{algorithm}
\begin{theorem}\label{thm:framework}  Let $\Phi$ be a continuous convex function with $[0,\infty)\subseteq \dom(\Phi)$, $m_L$ be an integer, $\mathcal{L}$ be a learning algorithm on $m_L$ samples, a function $S_L:\mathbb{R}\times\mathbb{R}\to \mathbb{N}$, and $\mathcal{P}$ be a set of joint distributions on $\X\times \Y$.  

Suppose the common prior between Alice and Bob satisfying Assumption~\ref{ass:apriori} with $P_{X,Y}\in \mathcal{P}$, and $\mathcal{L}$ is $(\delta, \epsilon)$-accurate on $(\mathcal{P}, \Phi)$ with $S_L(\delta, \epsilon)$ samples.  
Under three conditions in Assumption~\ref{ass:degenerate} respectively, Mechanism~\ref{alg:reduction} is
\begin{enumerate}
    \item  $(\delta, \epsilon)$-truthful on $\mathcal{P}$ with a $2+S_L(\delta, \epsilon)$ number of tasks;
    \item  $(\delta, \epsilon)$-informed-truthful on $\mathcal{P}$ with a $2+S_L(\delta, \epsilon)$ number of tasks;
    \item $(\delta, \epsilon)$-strongly truthful on $\mathcal{P}$ with a $2+S_L(\delta, \epsilon)$ number of tasks.
\end{enumerate}
\end{theorem}
Let $P\in \mathcal{P}$ be Alice and Bob's signals joint distribution.  Here $\mathcal{L}$ only outputs an $\epsilon$-ideal scoring function on the joint distribution of agents' signals. Still, the algorithm can have an arbitrarily large error when agents are not truthtelling.  For instance, there may exists a non-truth-telling strategy profile $\bm{\theta}$ such that $\bm{\theta}\circ P$ is not in $\mathcal{P}$, and the output of $\mathcal{L}$ is not $\epsilon$-ideal on $(\bm{\theta}\circ P,\Phi)$.  Nevertheless, Mechanism~\ref{alg:reduction} still can upper bound their ex-ante payment under such  non-truth-telling strategy profiles.  Furthermore, if the learning algorithm is $\epsilon$-ideal on $(\bm{\theta}\circ P,\Phi)$ for all strategy profile $\bm{\theta}$, the pairing mechanism is indeed approximately dominantly truthful. We give a more detail discussion in Sect.~\ref{sec:comparison}.
\begin{remark}\label{remark:framework}
Note that the truthfulness guarantees are subject to the belief of Alice (and Bob).  Mechanism~\ref{alg:reduction} ensures with $1-\delta$ probability the payment under truth-telling strategy profile is $\epsilon$ close to a fixed strongly truthful (inform-truthful or truthful) mapping for all $\delta\in (0,1)$ and $\epsilon>0$.\footnote{Formally, there exists an event with probability $1-\delta$ such that the conditional expected payment to Alice under such event is $\epsilon$-close to a strongly truthful (inform-truthful or truthful) mapping.}  In particular, we make the error $\epsilon$ sufficiently small such that the truth-telling strategy profile still has a higher ex-ante payment than any oblivious strategy has with high probability.

Furthermore, we can pick $\Phi$ such that the ex-ante payment is bounded by some constant $U$, and the mechanism is $(\epsilon+U\delta)$-strongly (informed-) truthful with probability $1$.  For example, if $\Phi(a) = |a-1|/2$, we only need to consider bounded scoring functions, and the resulting mechanism is approximately informed-truthful with probability $1$.  
\end{remark}
To establish some intuitions, let's consider the following ``fantasy'' mapping $F^\Phi = (F_A^\Phi, F_B^\Phi)$ from Alice's and Bob's signals' joint distribution $P$ and their strategy profile $\bm{\theta}$ to payments:
\begin{equation}\label{eq:fantasy}
    F_A^\Phi(\bm{\theta}, P) \triangleq u_A(\bm{\theta}, P, K^\star_{\bm{\theta}\circ P, \Phi})\text{ and } F_B^\Phi(\bm{\theta}, P) \triangleq u_B(\bm{\theta}, P, K^\star_{\bm{\theta}\circ P, \Phi}).
\end{equation}
It is straightforward to show the following lemma.
\begin{lemma}[Fantasy mapping]\label{lem:fantasy}
Under the first three conditions in Assumption~\ref{ass:degenerate} respectively, the mapping $F^\Phi = (F_A^\Phi, F_B^\Phi)$  is 
\begin{enumerate}
    \item truthful,
    \item informed-truthful, or
    \item strongly truthful.
\end{enumerate}
\end{lemma}

Recall that a learning algorithm for scoring function with input samples from $Q$ outputs an approximate ideal function $K^\star_{Q, \Phi}$.  If Alice and Bob have a strategy profile $\bm{\theta}$ with joint signal distribution $P$, the learning tasks are sampled from distribution $\bm{\theta}\circ P$ and a learning algorithm for scoring function will output an approximate version of $K^\star_{\bm{\theta}\circ P, \Phi}$.  Therefore, the ex-ante payment to Alice in Mechanism~\ref{alg:reduction} is ``close'' to fantasy payment function, and Theorem~\ref{thm:framework} formalizes this idea.  We show the proof in Appendix~\ref{sec:proof_framework}.

\section{Learning Ideal Scoring Functions}\label{sec:learning}
Theorem~\ref{thm:framework} reduces the mechanism design problem to a learning problem for an ideal scoring function.  However, Eqn.~\eqref{eq:approximated} may be hard to verify.  We provide two natural sufficient conditions for $\epsilon$-ideal scoring functions in Sect.~\ref{sec:sufficient}, and we will provide two concrete learning algorithms for scoring function in Sect.~\ref{sec:algorithm}. Finally, in Sect.~\ref{sec:nonexistence} we show an obstacle to designing exact strongly truthful, or inform-truthful mechanisms which use the $\Phi$-divergence-based method.

\subsection{Sufficient Conditions for Approximately \texorpdfstring{$\Phi$}{Phi}-Ideal Scoring Functions}\label{sec:sufficient}

\paragraph{Bregman divergence}
Given $a,b\in \R$ and a strictly convex and twice differentiable $\Phi:\R\to\R$, the standard Bregman divergence is
$\Phi(a)-\Phi(b)-\nabla\Phi(b)^\top(a-b)$.
It can be extended to \defn{Bregman divergence} between two functionals $f$ and $g$ over a probability space $(\Omega, \mathcal{F}, P)$~\cite{csiszar1995generalized}
$$B_{\Phi, P}(f,g) = \int \Phi(f(\omega))-\Phi(g(\omega))-\nabla\Phi(g(\omega))^\top(f(\omega)-g(\omega))dP(\omega).$$
\begin{lemma}[Bregman divergence and accuracy]\label{lem:bregman}
If $\Phi$ is strictly convex and twice differentiable on $[0,\infty)$,
$$D_\Phi(P_{X,Y}\| P_X P_Y)-u_{A}(\bm{\tau}, P, K) = B_{\Phi^*, P_X P_Y}(K,K^\star).$$
Therefore, if $B_{\Phi^*, P_X P_Y}(K,K^\star)\le \epsilon$, $K$ is an $\epsilon$-ideal scoring function on $(\Phi, P)$.
\end{lemma}
Since Bregman divergence capture an \emph{average distance} between a scoring function $K$ and the ideal one, if the scoring function $K$ is uniformly close to the ideal one $K^\star$, the Bergman divergence between $K$ and $K^\star$ is also small.

\paragraph{Total variation distance} On the other hand, we may first learn the prior  $P$ and compute an approximately ideal scoring function afterward.  This indirect method is also useful, because estimating the probability density function is a much well studied problem.

\begin{theorem}[Total variation to accuracy]\label{thm:generative}
Given $\Phi$ is a convex function and a prior $P_{X,Y}$ over  a finite space $\X\times\Y$, suppose there exist constants $0<\alpha<1$ and $c_L$ such that
\begin{align}
    &\forall x\in \X,y \in \Y,\; P_{X,Y}(x,y)>2\alpha\text{ or }P_{X,Y}(x,y) = 0,\label{eq:generative00}\\
    &\forall z,w\in [\alpha,1/\alpha],\;  |\Phi(z)-\Phi(w)|\le c_L|z-w|.\label{eq:generative01}
\end{align}
If $\|\hat{P}_{X,Y}-P_{X,Y}\|_{TV}\le \delta<\alpha$,\footnote{$\|\hat{P}-\hat{P}\|_{TV}= \sum_{\omega\in\Omega} |P(\omega)-\hat{P}(\omega)|$ is the total variation distance between $P$ and $\hat{P}$.} $\hat{K}(x,y)\in \partial \Phi\left(\frac{\hat{P}_{X,Y}}{\hat{P}_{X}\otimes\hat{P}_{Y}}\right)$ is a $\frac{6c_L}{\alpha^2}\delta$-ideal scoring function.
\end{theorem}
The first condition says the smallest nonzero probability $P_{X,Y}(x,y)$ is either constantly away from zero or equal to zero, and the second condition requires the function $\Phi$ is Lipschitz in $[\alpha,1/\alpha]$ which holds for all examples in Table~\ref{tab:fconjugate}.  With these conditions, if we have a good estimation $\hat{P}$ for $P$ with small total variation distance, we can compute a very accurate scoring function $\hat{K}$ from $\hat{P}$.  As we will see in Sect.~\ref{sec:algorithm}, the empirical distributions with $m_L$ samples satisfies this condition with high probability for large enough $m_L$.

\subsection{Learning Algorithms for Scoring Functions}\label{sec:algorithm}
\paragraph{Generative approach}
Recall that if $P$ is known, the ideal scoring function can be computed directly.  In a generative approach, we try to estimate the probability density function $P$ from reports and derive the scoring function afterward under the truth-telling strategy profile.  In general this generative approach is useful when $\mathcal{P}$ is on a finite space, or ${\mathcal{P}}$ is a parametric model by Theorem~\ref{thm:generative}.  Here we provide an example of a generative approach.

A standard way of learning probability density function is to use empirical distribution on $m_L$ samples (defined in Eqn.~\eqref{eq:empirical}).  The following theorem shows that the empirical distribution gives a good estimation in terms of total variation distance.

\begin{algorithm}[hbt!]
  \begin{algorithmic}[1]
    \REQUIRE A report profile $\hat{\mathbf{X}}_L, \hat{\mathbf{Y}}_L\in ( \X\times\Y)^{m_L}$ from learning tasks from Alice and Bob.
    \ENSURE A convex function $\Phi$ and its sub-gradient $\partial \Phi$
    \STATE Compute empirical distribution from Alice's and Bob's reports: for all events $E$ in $\X\times\Y$
$$\hat{P}_{X,Y}(E) = \frac{1}{m_L}\sum_{s = 1}^{m_L}\mathbb{I}[(\hat{x}_{s}, \hat{y}_{s})\in E],$$
and compute the marginal empirical distribution, for all events $E$ in $\X$ and $F$ in $\Y$
$$\hat{P}_X(E) = \frac{1}{m_L}\sum_{s = 1}^{m_L}\mathbb{I}[\hat{x}_{s}\in E], \text{ and }\hat{P}_Y(F) = \frac{1}{m_L}\sum_{s = 1}^{m_L}\mathbb{I}[\hat{y}_{s}\in F].$$
\STATE Compute the scoring function as
\begin{equation}\label{eq:alg_generative3}
     \begin{cases}\hat{K}(x,y)\in \partial \Phi\left(\frac{\hat{P}_{X,Y}(x,y)}{\hat{P}_X(x)\hat{P}_{Y}(y)}\right),&\mbox{ if } \hat{P}_X(x)\hat{P}_{Y}(y)\neq 0\\
    \hat{K}(x,y) = 0,&\mbox{ otherwise.}
    \end{cases}
\end{equation}
\caption{A generative algorithm}\label{alg:generative}
  \end{algorithmic}
\end{algorithm}
\begin{lemma}[Theorem 3.1 in \cite{devroye2012combinatorial}]\label{lem:boundtv} For all $\epsilon>0$, $\delta>0$, finite domain $\Omega$, distribution in $P$ in $\Delta_\Omega$, there exists $M = O\left(\frac{1}{\epsilon^2}\max(|\Omega|,\log(1/\delta))\right)$ such that for all $m_L\ge M$ the empirical distribution with $m_L$ i.i.d. samples, $\hat{P}_{m_L}$, satisfies
$$\Pr[\|P-\hat{P}_{m_L}\|_{TV}\le \epsilon]\ge 1-\delta.$$

\end{lemma}

Therefore, we can design a learning algorithm $\mathcal{L}_{\rm emp}$ as follows: estimate joint distribution $P_{X,Y}$ by their empirical distributions $\hat{P}_{X,Y}$ and derive $\hat{K}$ from Theorem~\ref{thm:generative}.  By Theorem~\ref{thm:generative} and Lemma~\ref{lem:boundtv}, such algorithm is $\epsilon$-accurate with $1-\delta$ probability.

\paragraph{Discriminative approach}
Instead of density estimation, a discriminative approach estimates an ideal scoring functions directly.  This enables more freedom of algorithm design.  Here we use the variational representation (Theorem~\ref{thm:fdivergenceconjugate}), and give an optimization characterization of an ideal scoring function.

Given the assumption~\ref{ass:apriori}, under the truth-telling strategy profile we can have i.i.d. samples of $(u,v)$ where $u$ is sampled from $P_{X,Y}$ and $v$ is sampled from $P_X P_Y$ independently, and this is shown formally in Algorithm~\ref{alg:empest}.  Taking $L^\Phi(a,b) \triangleq a-\Phi^*(b)$ as the risk function, we can convert the estimation of the ideal scoring functions to empirical risk minimization (maximization) over a training set $(u_t,v_t)$ with $t = 1,2,\ldots, \lfloor m_L/3\rfloor$,
\begin{equation}\label{eq:erm}
    \tilde{K} = \arg\max_{k\in \mathcal{K}}\sum_t L^\Phi(k(u_t),k(v_t)) = \arg\max_{k\in \mathcal{K}} \left\{\int k(\omega)d\hat{P}_{X,Y}(\omega)-\int \Phi^*(k(\omega))d\hat{P_X P_Y}(\omega)\right\}
\end{equation}
where $\mathcal{K}$ is a pre-specified class of functionals $k:\X\times \Y\to \R$,  $\hat{P}_{X,Y}$ and $\hat{P_X P_Y}$ are empirical distributions on $\lfloor m_L/3\rfloor$ samples from distributions $P_{X,Y}$ and $P_X P_Y$ respectively.  See Appendix~\ref{sec:pre_emp} for formal definitions when $\X$ and $\Y$ are general measure spaces.

Assuming that $\mathcal{K}$ is a convex set of functionals, the implementation of \eqref{eq:erm} only requires solving a convex optimization problem over function space $\mathcal{K}$ which is well studied \cite{nguyen2010estimating}.  With these results, we show \emph{the empirical risk maximizer $\tilde{{K}}$ with respect to $L^\Phi$ is $\epsilon$-accurate} with large probability under some conditions on $\mathcal{K}$ and prior $P_{X,Y}$.  Furthermore, this error can be seen as the generalized error of the empirical risk maximizer. 
\begin{algorithm}[hbt!]
\caption{An empirical risk minimization algorithm}\label{alg:empest}
  \begin{algorithmic}[1]
    \REQUIRE A report profile $\hat{\mathbf{X}}_L, \hat{\mathbf{Y}}_L\in ( \X\times\Y)^{m_L}$ from learning tasks from Alice and Bob.
    \ENSURE A convex function $\Phi$ and its conjugate $\Phi^*$.
    \STATE Partition the report profile into three equal size $(\hat{\mathbf{x}}^i,\hat{\mathbf{y}}^i)$ in $(\X\times \Y)^{m_L/3}$ where $i = 0,1$, and $2$.
    \STATE For the empirical joint distribution, we use the report profile $\mathbf{x}^i,\mathbf{y}^i$ to compute: For all events $E$ in $\X\times\Y$
$$    \tilde{P}_{X,Y}(E) \triangleq \frac{3}{m_L}\sum_{s = 1}^{m_L/3}\mathbb{I}[(\hat{x}^0_{s}, \hat{y}^0_{s})\in E],$$
Further compute the product empirical distributions: for all events $E$
$$    \tilde{P}_{i}\tilde{P}_{j}(A) \triangleq \frac{3}{m_L}\sum_{s = 1}^{m_L/3}\mathbb{I}[(\hat{x}^1_{s}, \hat{y}^2_{s})\in E]$$
(Note that we use new samples to compute the product of empirical distribution to ensure the independence between $\tilde{P}_{X}\tilde{P}_{Y}$ and $\tilde{P}_{X,Y}$)
\STATE Finally solve following optimization problem
\begin{equation}\label{eq:ermprime}
    \tilde{K} = \arg\max_{k\in \mathcal{K}} \left\{\int k(\omega)d\tilde{P}_{X,Y}(\omega)-\int \Phi^*(k(\omega))d\tilde{P}_X\tilde{P}_Y(\omega)\right\}
\end{equation}
  \end{algorithmic}
\end{algorithm}

\begin{theorem}\label{thm:empest}
Consider a distribution $P$ over $\X\times\Y$; a strictly convex and a twice differentiable function $\Phi$ on $[0,\infty)$ with its gradient $\Phi'$ and conjugate $\Phi^*$; a family of functional $\mathcal{K}$ from $\X\times\Y$ to $\dom (\Phi^*)$; and $\Phi^*(\mathcal{K}) = \{\Phi^*(k):k\in \mathcal{K}\}$.  Suppose
\begin{enumerate}
    \item the $(P, \Phi)$-ideal scoring function  $K^\star = \Phi'\left(\frac{P_{X,Y}}{P_X P_Y}\right)$ is in $\mathcal{K}$, and
    \item there exist constants $(L_l, R_l, D_l)_{l = 1,2}$
    \begin{enumerate}
        \item $\sup_{k\in\mathcal{K}} \rho_{L_1}(k,P_{X,Y})\le R_1$, and  $\int^{R_1}_0\sqrt{\mathcal{H}_{[],L_1}(u,\mathcal{K},P_{X,Y})} du\le D_1$
        \item $\sup_{l\in\Phi^*(\mathcal{K})}  \rho_{L_2}(l,P_{X}P_Y)\le R_2$ and  $\int^{R_2}_0\sqrt{\mathcal{H}_{[],L_2}(u,\Phi^*(\mathcal{K}),P_X P_Y)} du\le D_2$
    \end{enumerate}
\end{enumerate}
There exists $M = O\left(\frac{1}{\varepsilon^2}\log \frac{1}{\delta}\right)$, such that for all $m_L\ge M$, $\tilde{K}$ defined in \eqref{eq:ermprime} is $\varepsilon$-accurate on prior $P$ with probability $1-\delta$.  \footnote{Here we do not show the dependency on constants $L_l,R_l$ and $D_l$.}
\end{theorem}

Informally, Theorem~\ref{thm:empest} requires the functional class $\mathcal{K}$ contains an ideal scoring function and it has a constant complexity (generalized entropy with bracketing).  Under these conditions, the empirical risk minimizer (maximizer) can estimate the ideal scoring function accurately even when the signal space can be integers, real numbers, or Euclidean spaces.

Here we give a outline of the proof. By Lemma~\ref{lem:bregman}, it is sufficient to show the empirical risk minimizer $\tilde{K}$ has small Bregman divergence form the ideal one.  Moreover, if the estimation $K$ is the empirical risk maximizer, this error can be upper bounded by the distance between the empirical distribution and the real distribution (Lemma~\ref{lem:breg2emp}).  Therefore, we can use functional form of Central Limit Theorem to upper bound the error (Theorem~\ref{thm:uniform_emp}).  We defer the proof to the appendix.
\begin{lemma}\label{lem:breg2emp}
Let $\tilde{K}$ be the estimate of $K^\star$ obtained by solving Eqn.~\eqref{eq:ermprime}, and $K^\star\in \mathcal{K}$  Then
$$B_{\Phi^*, P_X P_Y}(\tilde{K},K^\star)\le \sup_{k\in\mathcal{K}}\left|\int \Phi^*(k- \Phi^*(K^\star)d(\tilde{P}_X\tilde{P}_Y-P_X P_Y)-\int \left(k-K^\star\right)d(\tilde{P}_{X,Y}-P_{X,Y}) \right|.$$
\end{lemma}
\begin{example}[continued]
For $\Phi(a) = a\log a$, if the parameters $\sigma^2, \tau^2$ are in a bounded set, we can take $\mathcal{K}$ as a set of the quadratic functions with bounded coefficients.  By Theorem 2.7.11~\cite{wellner2013weak}, the general bracket entropy of $\mathcal{K}$ and $\Phi^*(\mathcal{K})$ can be bounded by some constants.
\end{example}
\subsection{Nonexistence of Unbiased Estimators for \texorpdfstring{$\Phi$}{Phi}-divergence}\label{sec:nonexistence}
Combining Theorem~\ref{thm:framework} and Theorem~\ref{thm:generative} or \ref{thm:empest} we can design mechanisms that are $\epsilon$-strongly truthful (inform-truthful, or truthful) with high probability.  \emph{However, is it possible to have an exact informed-truthful or strongly truthful?}  In this section, we show a technical obstacle to designing such mechanisms.

The main observation of Theorem~\ref{thm:truth} and \ref{thm:framework} is that the ex-ante payment to an agent has a close connection to the $\Phi$-divergence from signal pairs on penalty tasks to signal pairs on the bonus task and use this $\Phi$-divergence to upper bound ex-ante payment under all manipulations uniformly.  This observation is also used in \cite{Radanovic2016-qr} and \cite{kong2019information}.  Under this framework, showing exact strongly truthful, informed-truthful, or truthful requires \emph{unbiased estimator} of $\Phi$-divergence from i.i.d. samples.  Specifically, suppose we can estimate an ideal scoring function accurately from samples.  We can estimate the $\Phi$-divergence without bias.  The following theorem shows such estimator does not exist in general.

\begin{theorem}[Nonexistence]\label{thm:nonexists}  Suppose the discrete signal spaces of Alice and Bob, $\X$ and $\Y$, both have more than two elements, and $\Phi$ be twice differentiable convex function in $[0,\infty)$.  For all $m\in \N$ all estimator $\hat{D}:(\X\times\Y)^{m}\to \R$ from $m$ pairs of signals $(\mathbf{x}, \mathbf{y}) = (x_1, y_1, \ldots, x_m,y_m)$ to a real value, there exists a prior distribution $P_{X,Y}$ over $\X\times\Y$ such that 
$$\E_{(\mathbf{x}, \mathbf{y})\sim P_{X,Y}^m}[\hat{D}(\mathbf{x}, \mathbf{y})]\neq D_\Phi(P_{X,Y}\|P_X P_Y).$$
\end{theorem}

The key idea of this proof is that if we fix the estimator $\hat{D}$ and take the probability distribution $P$ as variables, the expected value $\E[\hat{D}(\mathbf{x}, \mathbf{y})]$ is a polynomial of distribution $P$.  However, the $\Phi$-divergence $\E_{ P_X P_Y}\left[\Phi\left(\frac{P_{X,Y}}{P_X P_Y}\right)\right]$ is usually not a polynomial, and we can fine one $P_{X,Y}$ to make these two values not equal.  The proof is in Appendix~\ref{sec:proof_impossiblity}.

\section{Machine Learning and Multiple Agents}\label{sec:more}
We have discussed the $\Phi$-pairing mechanisms on two agents, Alice and Bob.  What can we do if there are more than two agents, Alice, Bob, \ldots?  We first discuss two naive approaches which reduce the multiple agents setting to the two agent setting.  Then we propose two novel approaches that exploit the power of current machine learning algorithms.

\paragraph{Two naive approaches}
First, we can pair Alice with a randomly chosen peer agent and run our mechanism.\footnote{Formally, suppose in agents' common prior each pair of agents' signals is from a stochastic relevant prior family $\mathcal{P}$, and the learning algorithm $\mathcal{L}^\Phi$ is $(\epsilon, \delta)$-accurate with $m_L$ samples over $\mathcal{P}$.  The above mechanism is $\epsilon$-strongly (informed) truthful with probability at least $1-\delta$.  For example, if all agents' signal are from a finite set $\Z$ and for any pair of agents their signals are stochastic relevant and satisfy Eqns.~\eqref{eq:generative1} and \eqref{eq:generative2}, then by Theorem~\ref{thm:generative} and Lemma~\ref{lem:boundtv}, for any $\epsilon, \delta>0$, there exists  $S(\delta, \epsilon) = O\left(\frac{c_L^2}{\alpha^4\epsilon^2} \cdot \max\left\{|\Z|^2,\log\frac{1}{\delta}\right\}\right)$, such that the above mechanism is $(\delta, \epsilon)$-strongly truthful, inform-truthful, or truthful) with $S(\delta, \epsilon)$ tasks.
The similar argument works for continuous signal by Theorem~\ref{thm:framework} and Theorem~\ref{thm:empest}.}  
This approach keeps the sample complexity low as the number of agents increases.  However, if the average quality of agents' reports decreases as the number of agents increases, Alice will receive less payment and may give up working.   For example, say only Alice and Bob work on the tasks and the rest of agents report random noise.  Alice will now only have positive expected payment if her randomly matched peer is Bob.  As the number of agents increases, her expected payment will go to zero.  

On the other hand, we can pair Alice simultaneously with all other agents, and run our mechanism.  As the number of agents increases, this approach ensures Alice's expected payment is non-decreasing.  This is because the mutual information does not decrease by adding more information---the additional agents' reports.  However, in the extreme example above, where only Alice and Bob work, the sample complexity for ideal scoring function will increase.

\paragraph{Computing the $\Phi$-Mutual Information between $X_i$ and $X_{-i}$}

A challenge to employing this second method is to reliably compute the $\Phi$-mutual information between Alice's reports, $X_i$, and those of the other agents, $X_{-i}$. Our variation method is well suited to this challenge.

Recall that Mechanism~\ref{alg:reduction} reduces the mechanism design problem to learning a scoring rule, which Eqn.~\eqref{eq:optimal} reduces to learning
$$\ratio_P(x_i,x_{-i}) = \frac{P_{X_i,X_{-i}}(x_i,x_{-i})}{P_{X_i}(x_i)P_Y(x_{-i})} = \frac{P_{X_i\mid X_{-i}}(x_i\mid x_{-i})}{P_{X_i}(x_i)}.$$

Therefore, it is enough to learn both the marginal distribution, $P_{X_i}(x_i)$ and $P_{X_i\mid X_{-i}}(x_i\mid x_{-i})$.   The former can be estimated empirically.  However, when the number of agents is large, the later is high dimensional and must be learned.  Fortunately, this is just a soft-classifier\footnote{That is, it produces a forecast to predict her report rather than a single report.} which, given the reports of every agent but Alice on a particular task, (soft) predicts Alice's report on the same task.

Therefore, we can derive an  approximate ideal scoring rule by using machine learning techniques to produce a (soft) prediction of Alice's report for an answer given the reports 
of the other agents.  Specifically, the machine learning algorithm outputs $f(\cdot, \cdot)$ such that $f(x_{i}, x_{-i}) = P_{X_i \mid X_{-i}}(x_i \mid x_{-i})$.

Using Mechanism~\ref{alg:reduction}, we can divide the tasks into training and testing tasks.  The training tasks are used to learn $f$ and to estimate $P_{X}(x)$.  We can compute $K_{\rm{est}}$ from $f$ and $P_X(x)$, and then use Mechanism~\ref{alg:reduction} to pay the agents.

Note that for the guarantees of Theorem~\ref{thm:framework} to hold, it is required that $f$ is learned accurately on truthful strategy profiles.  However, we do not require the learning algorithms perform well on non-truthful strategy profiles.

\paragraph{Latent Variable Models}

Our pairing mechanisms are particularly powerful when the prior $P$ on agents' signals is a latent variable model.  In a latent variable model, signals are mutually independent conditioned on the latent variables.  Examples include Dawid-Skene models, Gaussian mixture models, hidden Markov models, and latent Dirichlet allocations.  When $ P$ is a latent variable model, we can pay Alice the (approximate) mutual information between her report and each task's latent variable.  
\begin{enumerate}
    \item Given a latent label recovery algorithm e.g., \cite{Zhang2014-vl}, we run such algorithm on all reports except Alice's, and get estimate of latent label for each tasks $(Y_1, \ldots, Y_m)$;
    \item Then, using Alice's report $(X_1, \ldots, X_m)$ and the estimated latent label $(Y_1, \ldots, Y_m)$ we can run Mechanism~\ref{alg:reduction} and Algorithm~\ref{alg:generative} to pay Alice the mutual information between Alice's report and the latent labels.
\end{enumerate}
This mechanism is (approximate) strongly truthful, because the $\Phi$-mutual information between Alice and the others' reports is less than the $\Phi$-mutual information between her reports and the tasks' latent variable due to data processing inequality.
This approach has the following advantages.  First, this provides a reduction from aggregation to elicitation.  
Second, paying mutual information between Alice's reports and the latent variable resolves the problems that the above naive approaches have.  Alice's payment increases as the number of agents increases by the data processing inequality and the sample complexity of scoring function mirrors that of the latent label algorithm, which typically will not increase.

\section{Conclusion}
We showed how to reduce the design of peer prediction information elicitation in the multitask setting to a learning problem.  As a result, we extend multitask peer prediction to the continuous setting for parametric models with bounded learning complexity.  We also obtain improved bounds on the sample complexity for the finite signal setting. We note that in practice one could use deep learning techniques to learn the scoring function.  However, we leave it for future work to obtain rigorous bounds in this setting.

\bibliographystyle{plainnat}
\bibliography{main}
\newpage
\appendix
\section{Supplementary materials}
\subsection{Convex analysis}\label{sec:convex_add}
Here is a useful table for some standard $\Phi$s and their conjugate:
\begin{table}[h]
    \centering
    \caption{Common $\Phi$, its convex conjugate and subgradient}\label{tab:fconjugate}
\begin{tabular}{l l l l }
    \toprule
    $\Phi$-divergence & $\Phi(a)$ & $\Phi^*(b)$ & $\partial \Phi(a)$\\
    \hline
    Total variation   & $\frac{1}{2}|a-1|$ & $\begin{cases}b, &\text{if }|b|\le 1/2\\
    +\infty, &\text{otherwise}\end{cases}
    $ & $\begin{cases}1/2, &\text{if }a>1\\
    -1/2, &\text{if }a<1\\
    [-1/2,1/2] &\text{if } a = 1\end{cases}
    $\\
    KL-divergence   &   $a\log a$   &  $\exp(b-1)$ & $1+\log a$\\
    $\chi^2$-divergence &   $a^2-1$   &  $b^2/4+1$ & $2a$\\
    Squared Hellinger distance    &$\left(1-\sqrt{a}\right)^2$   &  $\begin{cases}b/(1-b), &\text{if }b<1\\
    -\infty, &\text{otherwise}\end{cases}$ & $1-1/\sqrt{a}$\\
    \bottomrule
\end{tabular}
\end{table}
Theorem~\ref{thm:fdivergenceconjugate} is a direct result of Young-Fenchel inequality:
\begin{theorem}[Young-Fenchel inequality]\label{thm:fenchel}
Given $a\in \dom(\Phi)$, for all $b\in \dom(\Phi^*)$, 
$$\Phi(a)\ge ab-\Phi^*(b),$$
where the equality holds when $b\in \partial \Phi(a) = \{d:\Phi(c)\ge \Phi(a)+\langle d, c-a\rangle\}$, and $b = \Phi'(a)$ if $\Phi$ is convex and differential at $a$.
\end{theorem}
\begin{proof}[Proof of Theorem~\ref{thm:fdivergenceconjugate}]
By the definition of $\Phi$-divergence,
\begin{align*}
    D_\Phi(P\|Q) =& \E_Q\left[\Phi\left(\frac{P}{Q}\right)\right]\\
    =& \E_Q\left[\sup_{b}\left\{b\cdot \frac{P}{Q}-\Phi^*(b) \right\}\right]\tag{by Young-Fenchel}\\
    =& \sup_{k:\Omega\to \dom(\Phi^*)}\left\{\E_Q\left[k(\omega)\cdot \frac{P(\omega)}{Q(\omega)}-\Phi^*(k(\omega)) \right]\right\}\\
    =& \sup_{k:\Omega\to \dom(\Phi^*)}\left\{\E_Q\left[k(\omega)\cdot \frac{P(\omega)}{Q(\omega)}\right]-\E_Q\left[\Phi^*(k(\omega)) \right]\right\}\\
    =& \sup_{k:\Omega\to \dom(\Phi^*)}\left\{\E_P\left[k(\omega)\right]-\E_Q\left[\Phi^*(k(\omega)) \right]\right\}
\end{align*}
Therefore, by Young-Fenchel inequality the equlity holds when $k(\omega)\in \partial \Phi(P(\omega)/Q(\omega))$ almost everywhere on $Q$.
\end{proof}

This formulation is powerful.  For example, it can yield the data processing inequality easily.

\begin{corollary}[Data processing inequality]\label{cor:dataprocessing}
Consider a channel that produces $Y$ given $X$ based on the distribution $P_{Y|X}$ where $\Pr[Y|X] = P_{Y|X}$.  Given distributions $P_X$ and $Q_X$ of $X$ and $P_{Y|X}$, $P_Y$ is the (marginal) distribution of $Y$ when $X$ is sampled from $P_X$ and $Q_Y$ is the distribution of $Y$ when $X$ is generated by $Q_X$, then for any $\Phi$-divergence $D_\Phi$,
$$D_\Phi(P_X\|Q_X)\ge D_\Phi(P_Y\|Q_Y).$$
\end{corollary}
\begin{proof}[Proof of Corollary~\ref{cor:dataprocessing}]
By Theorem~\ref{thm:fdivergenceconjugate}, there exists a real-valued function $g:\mathcal{Y}\to \R$ such that
\begin{align*}
    D_\Phi(P_Y\|Q_Y) =& \E_{P_Y}[g]-\E_{Q_Y}[\Phi^*(g)]\\
    =& \sum_{y\in \mathcal{Y}} P_Y(y)g(y)-\sum_{y\in \mathcal{Y}} Q_Y(y)\Phi^*(g(y))\\
    =& \sum_{x\in \mathcal{X}, y\in \mathcal{Y}} P_X(x)P_{Y|X}(y,x)g(y)-\sum_{x\in \mathcal{X}, y\in \mathcal{Y}} Q_X(x)P_{Y|X}(y,x)\Phi^*(g(y))\\
    =& \sum_{x\in \mathcal{X}} P_X(x)\sum_{y\in \mathcal{Y}}P_{Y|X}(y,x)g(y)-\sum_{x\in \mathcal{X}} Q_X(x)\sum_{y\in \mathcal{Y}}P_{Y|X}(y,x)\Phi^*(g(y)).
\end{align*}
Because $\Phi^*$ is convex and for all $x\in \mathcal{X}$, $P_{Y|X}(y,x)$ is a distribution over $y$, we have for all $x$ in $\mathcal{X}$, $
    \sum_{y\in \mathcal{Y}}P_{Y|X}(y,x)\Phi^*(g(y))\ge \Phi^*\left(\sum_{y\in \mathcal{Y}}P_{Y|X}(y,x)g(y)\right)$.
Therefore we have
$$    D_\Phi(P_Y\|Q_Y) \le  \sum_{x\in \mathcal{X}} P_X(x)\left(\sum_{y\in \mathcal{Y}}P_{Y|X}(y,x)g(y)\right)-\sum_{x\in \mathcal{X}} Q_X(x)\Phi^*\left(\sum_{y\in \mathcal{Y}}P_{Y|X}(y,x)g(y)\right).$$
Define $h(x)\triangleq \sum_{y\in \mathcal{Y}}P_{Y|X}(y,x)g(y)$, and we can further simplify it as
\begin{align*}
    D_\Phi(P_Y\|Q_Y) \le & \sum_{x\in \mathcal{X}} P_X(x)h(x)-\sum_{x\in \mathcal{X}} Q_X(x)\Phi^*\left(h(x)\right)\\
    \le& \sup_{h:\mathcal{X}\to \R} \sum_{x\in \mathcal{X}} P_X(x)h(x)-\sum_{x\in \mathcal{X}} Q_X(x)\Phi^*\left(h(x)\right)\\
    =& D_\Phi(P_X\|Q_X)
\end{align*}
which completes the proof.
\end{proof}
\subsection{Upper bounds for empirical processes}\label{sec:pre_emp}
In this section, we provide some standard results on empirical process, most of which are in \citet{van2000empirical}.  Consider $n$ independent and identically (i.i.d.) random variables $X_1, X_2, \ldots, X_n$ with distribution $P$ on a measurable space $(\Omega, \mathcal{F})$.  Let $\hat{P}_n$ be the \emph{empirical distribution} based on those $n$ random variables, i.e., for each set $A\in\mathcal{F}$,
\begin{equation}\label{eq:empirical}
    \hat{P}_n(A) = \frac{1}{n}\{\text{number of }X_i\in A, 1\le i\le n\} = \frac{1}{n}\sum_{i = 1}^n\mathbb{I}[X_i\in A]
\end{equation}
and let $\mathcal{K}\subset L_2(P) = \{k:\Omega\to\R:\int |k|^2dP<\infty\}$ be a collection of functions.  The \emph{empirical process indexed} by $\mathcal{K}$ is
\begin{equation}\label{eq:empirical_proc}
    V_n(\mathcal{K}) = \left\{v_n(k) = \sqrt{n}\int k d(\hat{P}_n-P):k\in \mathcal{K}\right\}.
\end{equation}
In this paper, we are mainly interested in uniform upper bound for Eqn.~\eqref{eq:empirical_proc}, i.e., upperbounds for
\begin{equation}\label{eq:empirical_radius}
    \sup_{k\in\mathcal{K}}|v_n(k)|
\end{equation} which can be think as the ``radius'' of random process \eqref{eq:empirical_proc}.
To upper bound \eqref{eq:empirical_radius}, there are several notions for ``complexity of functional spaces''.  Here are some examples.  If $\Omega\subseteq \R$, the set of cumulative density functions is $\{k_x:k_x(\omega) = \mathbb{I}[\omega<x], x\in \R\}$, and the upperbound for \eqref{eq:empirical_radius} implies the Central Limit Theorem.  We can consider a family of sets $\mathcal{A}\subseteq \mathcal{F}$ and a functional class over it $\{k_A:k_A(\omega) = \mathbb{I}[\omega\in A], A\in \mathcal{A}\}$, and the upper bound for \eqref{eq:empirical_radius} can be characterized by the \emph{VC-dimension} of the family of sets $\mathcal{A}$.  Or if $P$ is a distribution $d$-dimensional Gaussian and there is a set of linear functional $\{k_v:k_v(\omega) = v^\top \omega, \|v\|_2\le 1\}$, we can use \emph{metric entropy} to encode their complexity.

Now let us introduce some notions of functional complexity we used in the paper.
\begin{definition}
Given $k\in \mathcal{K}$, $L>0$, and distribution $P$, we define
$$\rho_L^2(k,P)\triangleq 2L^2\int \exp\left(\frac{|k|}{L}\right)-1-\frac{|k|}{L} dP$$
the \emph{Bernstein difference} between $k_1$ and $k_2$ is then $\rho_L^2(k_1-k_2,P)$ which can be seen as an extension of $L_2(P)$-norm, because $2(e^x-1-x)\approx x^2$ when $x$ is small.
\end{definition}

\begin{definition}[Generalized entropy with bracketing]  Let  $\mathcal{N}_{[],L}(\delta, \mathcal{L}, P)$ be the smallest value of $n$ for which there exists $n$ pairs of functions $\{(k_j^L, k_j^U)\}$ such that $\rho_L(k_j^U-k_j^L,P)\le \delta$ for all $j = 1, \ldots, n$ and such that for all $k\in \mathcal{K}$ there is a $j$ such that for all $\omega\in \Omega$
$$k_j^L(\omega)\le k(\omega)\le k_j^U(\omega).$$
Then $\mathcal{H}_{[],L}(\delta, \mathcal{K}, P) = \log \mathcal{N}_{[],L}(\delta, \mathcal{K}, P)$ is called the \emph{generalized entropy with bracketing}.
\end{definition}

A useful application of bracketing is to classes of parametric functions $\{k_t:t\in T\}$ that are Lipschitz in the parameter $t\in T$: There exists a metric $d$ on $T$ and a function $F:\Omega\to \mathbb{R}$ such that
$$|k_t(w)-f_s(w)|\le d(s,t) F(w)\text{ for all } w\in \Omega$$
Then the bracketing numbers of this class are bounded by the covering numbers of $T$.
\begin{theorem}
Let $\mathcal{K}_T = \{k_t:t\in T\}$ be a set of function satisfying the above condition.  Then for any norm $\|\cdot \|$ and $\epsilon>0$,
$$\mathcal{N}_{[]}(2\epsilon\|F\|, \mathcal{K}, \|\cdot\|)\le \mathcal{N}(\epsilon, T, d).$$
\end{theorem}
The following theorem shows the random variable \eqref{eq:empirical_radius} is subgaussian when the generalized entropy with bracketing is bounded.
\begin{theorem}[A uniform inequality~\cite{van2000empirical}]\label{thm:uniform_emp}
Given a functional class $\mathcal{K}$ and distribution $P$, if there exist constants $L$, $R$, $A$, $B$, and $C$ such that $\sup_{k\in\mathcal{K}} \rho_L(k,P)\le R$ and $\int^R_0\sqrt{\mathcal{H}_{[],L}(u,\mathcal{K},P)} du\le C$
$$\sqrt{(A+1)B^2}\left(\max\left\{ R, C\right\}\right)\le
    \epsilon\le \frac{A R^2}{L}\sqrt{n}$$
Then the empirical process $V_n(\mathcal{K})$ is bounded as
$$\Pr\left[\sup_{k\in\mathcal{K}}|v_n(k)|\ge \epsilon\right]\le  B\exp\left(-\frac{\epsilon^2}{(A+1)B^2R^2}\right).$$
\end{theorem}

\section{Proofs in Sect.~\ref{sec:lem}}\label{app:lem}

\manipulation*

There are two aspects of manipulation: the reports for bonus and penalty tasks and the scoring functions $K$.  The first one can be handled through Data Processing inequality and the second is shown through the variational representation of $\Phi$-divergence.

\begin{proof}
The expected utility for Alice is
\begin{align*}
    &u_{A}(\bm{\theta},P,K)\\
    =& \E_{\mathbf{X}, \mathbf{Y}}\left[\E_{\bm{\theta}}\left[K\left(\hat{X}_{b}, \hat{Y}_{b}\right)\mid \mathbf{X}, \mathbf{Y}\right]\right]-\E_{\mathbf{X}, \mathbf{Y}}\left[\E_{\bm{\theta}}\left[\Phi^*\left(K\left(\hat{X}_{p}, \hat{Y}_{q}\right)\right)\mid \mathbf{X}, \mathbf{Y}\right]\right]\\
    =& \E_{\mathbf{X}, \mathbf{Y}}\left[\sum_{\hat{x}, \hat{y}}\theta_A(X_{b},\hat{x})\theta_B(Y_{b},\hat{y})K(\hat{x},\hat{y})\right]-\E_{\mathbf{X}, \mathbf{Y}}\left[\sum_{\hat{x}, \hat{y}}\theta_A(X_{p},\hat{x})\theta_B(Y_{q},\hat{y})\Phi^*\left(K(\hat{x},\hat{y})\right)\right]\\
    =& \sum_{x,y}P_{X,Y}(x,y)\sum_{\hat{x}, \hat{y}}\theta_A(x,\hat{x})\theta_B(y,\hat{y})K(\hat{x},\hat{y})-\sum_{x,y}P_X(x)P_Y(y)\sum_{\hat{x}, \hat{y}}\theta_A(x,\hat{x})\theta_B(y,\hat{y})\Phi^*\left(K(\hat{x},\hat{y})\right).
\end{align*}
The last equality uses the fact that $P_{X,Y}$ is the joint distribution of signals on bonus task $b$, $(x_{b},y_{b})$ , and $P_X P_Y$ is the joint distribution of signals on penalty tasks $p$ and $q$, $(x_{p},y_{q})$ due to Assumption~\ref{ass:apriori}. 

Because $\Phi^*$ is convex and fixing $x$ and $y$, $\theta_A(x,\hat{x})\theta_B(y,\hat{y})$ is a distribution over $\X\times\Y$, by Jensen's inequality we have
\begin{equation}\label{eq:strict1}
    \sum_{\hat{x},\hat{y}}\theta_A(x,\hat{x})\theta_B(y,\hat{y})\Phi^*\left(K(\hat{x},\hat{y})\right)\le \Phi^*\left(\sum_{\hat{x}, \hat{y}}\theta_A(x,\hat{x})\theta_B(y,\hat{y})K_{i,B}(\hat{x},\hat{y})\right)
\end{equation}
where the equality holds only if $\Phi^*$ is not strictly convex or $K(\hat{x},\hat{y})$ is constant in the support of $(\hat{x}, \hat{y})\mapsto \theta_A(x,\hat{x})\theta_B(y,\hat{y})$. Let $L(x,y)\triangleq \sum_{\hat{x}, \hat{y}}\theta_A(x,\hat{x})\theta_B(y,\hat{y})K(\hat{x},\hat{y})$.  Apply Eqn.~\eqref{eq:strict1} to $u_{A}$ and  we have
\begin{align}
    u_{A}\le& \sum_{x,y}P_{X,Y}(x,y)L(x,y)-\sum_{x,y}P_X(x)P_Y(y)\Phi^*\left(L(x,y)\right)\nonumber\\
    \le& \sup_{K:\X\times\Y\to \R}\left\{\sum_{x,y}P_{X,Y}(x,y)K(x,y)-\sum_{x,y}P_X(x)P_Y(y)\Phi^*\left(K(x,y)\right)\right\}\label{eq:strict2}\\
    =& D_\Phi(P_{X,Y}\|P_X P_Y).\nonumber
\end{align}
The last inequality holds by Theorem~\ref{thm:fdivergenceconjugate}, and it completes the proof.
\end{proof}

\truthtelling*
\begin{proof}[Proof Lemma~\ref{lem:truthtelling}]
By Theorem~\ref{thm:fdivergenceconjugate}, and the definition of Truth-telling strategy profile, we have
\begin{align*}
    &u_{A}(\bm{\tau},P,K^\star)\\
    =& \E_{\mathbf{X},\mathbf{Y}}\left[\E_{\bm{\tau}}\left[K^\star\left(\hat{X}_b, \hat{Y}_b\right)\mid \mathbf{X}, \mathbf{Y}\right]\right]-\E_{\mathbf{X},\mathbf{Y}}\left[\E_{\bm{\tau}}\left[\Phi^*\left(K^\star\left(\hat{X}_p, \hat{Y}_q\right)\right)\mid \mathbf{X}, \mathbf{Y}\right]\right]\\
    =& \E_{\mathbf{X},\mathbf{Y}}\left[K^\star(X_b, Y_b)\right]-\E_{\mathbf{X},\mathbf{Y}}\left[\Phi^*\left(K^\star(X_p, Y_q)\right)\right]\tag{definition of $\bm{\tau}$}\\
    =& \sum_{x,y}P_{X,Y}(x,y)K^\star(x,y)-\sum_{x,y}P_X(x)P_Y(y)\Phi^*\left(K^\star(x,y)\right)\\
    =& \sup_{K:\X\times\Y\to \R}\left\{\sum P_{X,Y}K-\sum P_X P_Y\Phi^*\left(K\right)\right\}\tag{by Theorem~\ref{thm:fdivergenceconjugate} and $K^\star$}\\
    =& D_\Phi(P_{X,Y}\|P_X P_Y).
\end{align*}
Moreover, because $P_{X,Y}$ is stochastic relevant, $D_\Phi(P_{X,Y}\|P_X P_Y)>0$.
\end{proof}

\oblivious*
\begin{proof}[Proof of Lemma~\ref{lem:oblivious}]
Recall that an oblivious strategy $\theta_A$  is oblivious to the private signal: for any $x$ ,$x'$ and $\hat{x}$ in $\X$,  $\theta_A(x,\hat{x})=\theta_A(x',\hat{x})$, and we can define a distribution $\mu_A\in \Delta_{\X}$ such that for all $x$ and $\hat{x}$ in $\X$, $\mu_A(\hat{x}) \triangleq \theta_A(x,\hat{x})$.  We also define $\nu_B(\hat{y}) \triangleq \sum_{y} P_Y(y)\theta_B(y, \hat{y})$ where $\nu_B$ is a distribution on $\Y$ and independent to $\mu_A$.

\begin{align*}
    &u_{A}(\bm{\theta},P,K)\\
    =& \E_{\mathbf{X},\mathbf{Y}}\left[\E_{\bm{\theta}}\left[K\left(\hat{X}_b, \hat{Y}_b\right)\mid \mathbf{X},\mathbf{Y}\right]\right]-\E_{\mathbf{X},\mathbf{Y}}\left[\E_{\bm{\theta}}\left[\Phi^*\left(K\left(\hat{X}_p, \hat{Y}_q\right)\right)\mid \mathbf{X},\mathbf{Y}\right]\right]\\
    =& \E_{\mathbf{X},\mathbf{Y}}\left[\sum_{\hat{x}, \hat{y}}\theta_A(X_{b},\hat{x})\theta_B(Y_{b},\hat{y})K(\hat{x},\hat{y})\right]-\E_{\mathbf{X},\mathbf{Y}}\left[\sum_{\hat{x}, \hat{y}}\theta_i(X_{p},\hat{x})\theta_B(Y_{q},\hat{y})\Phi^*\left(K(\hat{x},\hat{y})\right)\right]\\
    =& \E_{\mathbf{X},\mathbf{Y}}\left[\sum_{\hat{x}, \hat{y}}\mu_A(\hat{x})\theta_B(Y_{b},\hat{y})K(\hat{x}, \hat{y})\right]-\E_{\mathbf{X},\mathbf{Y}}\left[\sum_{\hat{x}, \hat{y}}\mu_A(k)\theta_B(Y_{p},\hat{y})\Phi^*\left(K(\hat{x}, \hat{y})\right)\right]\\
    =& \sum_{\hat{x}, \hat{y}}\mu_A(\hat{x})\nu_B(\hat{y})\left[K(\hat{x}, \hat{y})-\Phi^*\left(K(\hat{x}, \hat{y})\right)\right]\tag{by the definition of $\nu_B$}\\
    \le& \sum_{\hat{x}, \hat{y}}\mu_A(\hat{x})\nu_B(\hat{y})\left[\sup_{b\in \dom(\Phi^*)}\left\{1\cdot b-\Phi^*\left(y\right)\right\}\right] = \sup_{b\in \dom(\Phi^*)}\left\{1\cdot b-\Phi^*\left(y\right)\right\}\\
    =& \Phi^{**}(1) = \Phi(1) = 0.
\end{align*}
The last inequality is from the Definition~\ref{def:fdiv}.
\end{proof}

\eqauality*
\begin{proof}[Proof of Lemma~\ref{lem:strict}]

We first prove the first property: $\theta_A$ and $\theta_B$ are permutations.  Note that by the proof of Lemma~\ref{lem:manipulation}, $u_{A}(\bm{\theta}, P, K) = D_\Phi(P_{i,j}\| P_{i}P_{j})$ if and only if \eqref{eq:strict1} and \eqref{eq:strict2} are equalities, because $\X$ and $\Y$ are finite

Given Alice's and Bob's strategies $\theta_A$ and $\theta_B$, let $S_A(x)\triangleq\{\hat{x}\in \X:\theta_A(x,\hat{x})>0\}$, $S_B(y)\triangleq\{\hat{y}\in \Y:\theta_B(y,\hat{y})>0\}$ be the support of strategy $\theta_A$ on signal $x$ and $\theta_B$ on $y$ respectively.  Because $\Phi^*$ is strictly convex and $\Phi$ is differentiable, $u_{A}(\bm{\theta}, P, K) = D_\Phi(P_{X,Y}\| P_X P_Y)$ if and only if \eqref{eq:strict1}, and \eqref{eq:strict2} are equality.  Thus
\begin{equation}\label{eq:equal1}
    \forall x\in\X,y\in \Y, \hat{x}\in S_A(x), \hat{y}\in S_B(y),\; K(\hat{x},\hat{y}) = \Phi'\left(\frac{P_{X,Y}(x,y)}{P_{X}(x)P_{Y}(y)}\right).
\end{equation}
That is all reports pairs $(\hat{x},\hat{y})$ in the support of strategy $\theta_A$ on $x$ and $\theta_B$ on $y$ have the same score, $K(\hat{x}, \hat{y})$.  Moreover, the value equals to $\Phi'\left(P_{X,Y}(x,y)/(P_{X}(x)P_{Y}(y))\right)$.  Now we use this observation to finish the proof.

$\Rightarrow$) Because $\theta_A(x,\cdot)$ induces a probability, $|S_A(x)|\ge 1$ for all $x$.  Suppose $\theta_A$ is not a permutation. Because $\X$ is finite, there exists $x_1\neq x_2$ and $\hat{x}^*$ in $\X$ such that $\hat{x}^*\in S_A(x_1)$ and $\hat{x}^*\in S_A(x_2)$.  By \eqref{eq:equal1}, for all $y$ and $\hat{y}\in S_B(y)$,
$$\Phi'\left(\frac{P_{X,Y}(x_1,y)}{P_{X}(x_1)P_{Y}(y)}\right) = K(\hat{x}^*,\hat{y}) = \Phi'\left(\frac{P_{X,Y}(x_2,y)}{P_{X}(x_2)P_{Y}(y)}\right).$$
Because $\Phi$ is strictly convex and differentiable, $\Phi'$ is invertible, and thus for all $y\in \Y$, \begin{align*}
    \frac{P_{X,Y}(x_1,y)}{P_{X}(x_1)P_{Y}(y)} = \frac{P_{X,Y}(x_2,y)}{P_{X}(x)P_{Y}(y)}
\end{align*}
which shows $P_{X,Y}$ is not stochastic relevant--- Given signal $x_1$ Alice's poster for Bob's signal is identical to her poster with signal $x_2$--- and reaches contradiction.  Therefore there exist permutations $\pi_A$ and $\pi_B$ over $\X$ and $\Y$ such that $\theta_A = \pi_A$ and $\theta_B = \pi_B$.  

For the second part, by \eqref{eq:equal1}, for all $x,y$ we have
$$K(\pi_A(x),\pi_B(y)) = \Phi'\left(\frac{P_{X,Y}(x,y)}{P_{X}(x)P_{Y}(y)}\right).$$

$\Leftarrow$) On the other hand, if $\theta_A = \pi_A$ and $\theta_B = \pi_B$ which are permutations, and for all $x,y$, and  $K(\pi_A(x),\pi_B(y)) = \Phi'\left(P_{X,Y}(x,y)/(P_{X}(x)P_{Y}(y))\right)$, we can apply Eqn.~\eqref{eq:equal1}, and have $u_{A}(\bm{\theta}, P, K) = D_\Phi(P_{X,Y}\| P_X P_Y)$.
\end{proof}

\section{Proofs in Sect.~\ref{sec:framework}}\label{sec:proof_framework}
\begin{proof}[Proof of Lemma~\ref{lem:fantasy}]
Given a prior $P$, the payment to Alice under truth-telling strategy profile in the fantasy function~\eqref{eq:fantasy} by Lemma~\ref{lem:truthtelling} is
\begin{equation}\label{eq:fantasy1}
    F^\Phi_A(\bm{\tau}, P) = u_A(\bm{\tau}, P, K^\star_{P, \Phi}) = D_\Phi(P_{X,Y}\| P_X P_Y).
\end{equation}

Additionally, by Lemma~\ref{lem:manipulation}, 
$$F^\Phi_A(\bm{\theta}, P) = u_A(\bm{\theta}, P, K^\star_{\bm{\theta}\circ P, \Phi}) \le D_\Phi(P_{X,Y}\| P_X P_Y) = F^\Phi_A(\bm{\tau}, P)$$
which shows the truth-telling strategy profile is a Bayesian Nash equilibrium.  

To show that the mapping is inform-truthful, by Lemma~\ref{lem:oblivious}, if $\bm{\theta}$ is an oblivious strategy profile, 
$F^\Phi_A(\bm{\theta}, P) = u_A(\bm{\theta}, P, K^\star_{\bm{\theta}\circ P, \Phi}) \le 0$.
Therefore when $P$ is stochastic relevant by Eqn.~\eqref{eq:fantasy1} and Lemma~\ref{lem:truthtelling} we have 
$$F^\Phi_A(\bm{\theta}, P)\le 0 <F^\Phi_A(\bm{\theta}, P).$$

Finally, to show the mechanism is strongly truthful, if there is a strategy profile $\bm{\theta}$ such that $F^\Phi_A(\bm{\theta}, P) = F^\Phi_A(\bm{\tau}, P)$, we have
$$u_A(\bm{\theta}, P, K^\star_{\bm{\theta}\circ P, \Phi}) = u_A(\bm{\tau}, P, K^\star_{P, \Phi}),$$
so by Lemma~\ref{lem:strict} $\bm{\theta}$ is a permutation strategy profile which completes the proof.
\end{proof}

Note that the statement of Theorem~\ref{thm:framework} is a little subtle.  Mentioned in the footnote in Remark~\ref{remark:framework} there the statement consists of two parts of randomness: an event with probability $1-\delta$, and the conditional expected payment to Alice under such event is $\epsilon$-close to a strongly truthful (inform-truthful or truthful) mapping.   Therefore, to prove Theorem~\ref{thm:framework}, it is sufficient to show there exists an event $\mathcal{E}$ such that 
\begin{enumerate}
    \item it happens with probability at least $1-\delta$, 
    \item Alice's conditional ex-ante utility under truth-telling strategy profile is $\epsilon$-close to $F^\Phi_A(\bm{\tau}, P)$ defined in \eqref{eq:fantasy}, and  
    \item for all strategy profile $\bm{\theta}$ Alice's conditional ex-ante utility under $\bm{\theta}$ is less than $F^\Phi_A(\bm{\theta}, P)$.
\end{enumerate}
\begin{proof}[Proof of Theorem~\ref{thm:framework}]
First, if Alice's and Bob's strategy profile is $\bm{\theta}$, the learning tasks $(\hat{\mathbf{x}}_L,\hat{\mathbf{y}}_L)$ and scoring tasks $(\hat{\mathbf{x}}_S,\hat{\mathbf{y}}_S)$ are both generated from distribution $\bm{\theta}\circ P$ i.i.d..  Additionally, the ex-ante payment to Alice is over two randomness: learning tasks and scoring tasks.  To make this distinction explicit, we let $U_A(\hat{\mathbf{x}}_S,\hat{\mathbf{y}}_S, \hat{\mathbf{x}}_L,\hat{\mathbf{y}}_L)$ be Alice's payment when the report profile is $(\hat{\mathbf{x}},\hat{\mathbf{y}}) = (\hat{\mathbf{x}}_S,\hat{\mathbf{y}}_S, \hat{\mathbf{x}}_L,\hat{\mathbf{y}}_L)$.  Then Alice's ex-ante payment under strategy profile $\bm{\theta}$ in mechanism $\mathcal{M}^{\Phi, \mathcal{L}}$ is 
$$u_A(\bm{\theta}; P, \mathcal{M}^{\Phi, \mathcal{L}}) = \E_{\substack{(\hat{\mathbf{x}}_L,\hat{\mathbf{y}}_L)\sim \bm{\theta}\circ P^{m_L};\\ (\hat{\mathbf{x}}_S,\hat{\mathbf{y}}_S)\sim {\bm{\theta}\circ P^{m_S}}}}\left[U_A(\hat{\mathbf{x}}_S,\hat{\mathbf{y}}_S, \hat{\mathbf{x}}_L,\hat{\mathbf{y}}_L)\right] = \E_{(\hat{\mathbf{x}}_L,\hat{\mathbf{y}}_L)}\left[\E_{(\hat{\mathbf{x}}_S,\hat{\mathbf{y}}_S)}\left[U_A(\hat{\mathbf{x}}_S,\hat{\mathbf{y}}_S, \hat{\mathbf{x}}_L,\hat{\mathbf{y}}_L)\mid (\hat{\mathbf{x}}_L,\hat{\mathbf{y}}_L)\right]\right].$$
We further define 
$$u_A(\bm{\theta}, P, \mathcal{L}(\mathbf{x}_L, \mathbf{y}_L)) \triangleq \E_{(\hat{\mathbf{x}}_S,\hat{\mathbf{y}}_S)\sim \bm{\theta}\circ P}\left[U_A(\hat{\mathbf{x}}_S,\hat{\mathbf{y}}_S, \mathbf{x}_L,\mathbf{y}_L)\right]$$
where the expectation is only taken on the scoring tasks, but the learning tasks are fixed.

Now we define an event 
$$\mathcal{E} = \{\mathbf{x}_L, \mathbf{y}_L):u_{A}(\bm{\tau}, P, \mathcal{L}(\mathbf{x}_L, \mathbf{y}_L)) > D_\Phi(P_{X,Y}\| P_X P_Y)-\epsilon\}$$
which is in the probability space generated by the learning tasks.  Because $\mathcal{L}$ is $( \delta, \epsilon)$-accurate on $\mathcal{P}$ and the joint signal distribution $P\in \mathcal{P}$, the probability of $\mathcal{E}$ is greater than $1-\delta$.

By the definition of $\mathcal{E}$, for all $(\mathbf{x}_L, \mathbf{y}_L)\in \mathcal{E}$,
\begin{align*}
    &u_{A}(\bm{\tau}, P, \mathcal{L}(\mathbf{x}_L, \mathbf{y}_L))\\
    >& D_\Phi(P_{X,Y}\| P_X P_Y)-\epsilon\\
    =& u_A(\bm{\tau}, P, K^\star_{\bm{\tau}\circ P, \Phi})-\epsilon\tag{by Lemma~\ref{lem:truthtelling}}\\
    =&F_A^\Phi(\bm{\tau}, P)-\epsilon\tag{by \eqref{eq:fantasy}}
\end{align*}
Therefore, Alice's conditional expected payment under truth-telling strategy profile is $\epsilon$ close to the fantasy function.

Finally, it is sufficient that for all $(\mathbf{x}_L, \mathbf{y}_L)$, and $\bm{\theta}$
$$u_{A}(\bm{\theta}, P, \mathcal{L}(\mathbf{x}_L, \mathbf{y}_L))\le F_A^\Phi(\bm{\theta}, P).$$
Formally, for all $P\in \mathcal{P}$, $\bm{\theta}$, and $(\mathbf{x}_L, \mathbf{y}_L)$, 
\begin{align*}
    &u_{A}(\bm{\theta}, P, \mathcal{L}(\mathbf{x}_L, \mathbf{y}_L))\\
    =&u_{A}(\bm{\tau}, \bm{\theta}\circ P, \mathcal{L}(\mathbf{x}_L, \mathbf{y}_L))\tag{by \eqref{eq:associative}}\\
    \le& u_A(\bm{\tau}, \bm{\theta}\circ P, K^\star_{\bm{\theta}\circ P, \Phi})\tag{by Lemma~\ref{lem:manipulation}}\\
    =& u_A(\bm{\theta}, P, K^\star_{\bm{\theta}\circ P, \Phi})\tag{by \eqref{eq:associative}}\\
    =& F_A^\Phi(\bm{\theta}, P),
\end{align*}
and we complete the proof.
\end{proof}

\section{Proofs in Sect.~\ref{sec:learning}}

\subsection{Proofs in Sect.~\ref{sec:sufficient}}

\begin{proof}[Proof of Lemma~\ref{lem:bregman}]
\begin{align*}
    &D_\Phi(P_{X,Y}\| P_X P_Y)-u_{A}(\bm{\tau},P,K)\\
    =& \int K^\star dP_{X,Y} - \int \Phi^*(K^\star)dP_X P_Y-\int KdP_{X,Y} + \int \Phi^*(K)dP_X P_Y\\
    =& \int \Phi^*(K)- \Phi^*(K^\star)+\frac{dP_{X,Y}}{dP_X P_Y}\left(K^\star-K\right) dP_X P_Y\\
    =& \int \Phi^*(K)- \Phi^*(K^\star)-(\Phi^*)'\left(K^\star\right)\left(K-K^\star\right) dP_X P_Y
\end{align*}
The last equality holds since $K^\star(x,y) = \Phi'\left(\frac{dP_{X,Y}(x,y)}{dP_X P_Y(x,y)}\right)$, so $(\Phi^*)'\left(K^\star(x,y)\right) = \frac{dP_{X,Y}(x,y)}{dP_X P_Y(x,y)}$ by Theorem~\ref{thm:fenchel}.  
The final line is indeed the Bregman divergence from $K^\star$ to $K$ with respect to measure $dP_X P_Y$ and $\Phi^*$.
\end{proof}

\begin{proof}[Proof of Theorem~\ref{thm:generative}]
Let $\hat{{K}}$ be the output of Algorithm~\ref{alg:generative}, and $K^\star$ be a $(\Phi, P)$-ideal scoring function defined in \eqref{eq:optimal}.  We have
\begin{align*}
    &u_{A}(\bm{\tau},P,\hat{{K}})\\
    =& \E_{\mathbf{X},\mathbf{Y}}\left[\hat{K}(X_{b},Y_{b})\right]-\E_{\mathbf{X}, \mathbf{Y}}\left[\Phi^*\left(\hat{K}(X_{p},Y_{q})\right)\right]\\
    =& \sum_{x,y} P_{X,Y}(x,y)\hat{K}(x,y)-P_{X}(x)P_Y(y)\Phi^*\left(\hat{K}(x,y)\right)\\
    =& \sum_{x,y :P_{X}(x)P_Y(y)\neq 0}P_{X}(x)P_Y(y)\left[ \frac{P_{X,Y}(x,y)}{P_{X}(x)P_Y(y)}\hat{K}(x,y)-\Phi^*\left(\hat{K}(x,y)\right)\right]\\
    =& \sum P_{X}P_Y\left[ \frac{\hat{P}_{X,Y}}{\hat{P}_{X}\hat{P}_Y}\hat{K}-\Phi^*\left(\hat{K}\right)+\left(\frac{P_{X,Y}}{P_{X}P_Y}-\frac{\hat{P}_{X,Y}}{\hat{P}_{X}\hat{P}_Y}\right)\hat{K}\right]
\end{align*}
Because $\hat{K}(x,y) \in \partial \Phi\left(\frac{\hat{P}_{X,Y}(x,y)}{\hat{P}_X(x)\hat{P}_{Y}(y)}\right)$, by Young-Fenchel inequality (Theorem~\ref{thm:fenchel}) we have 
$ \frac{\hat{P}_{X,Y}(x,y)}{\hat{P}_X(x)\hat{P}_Y(y)}\hat{K}(x,y)-\Phi^*\left(\hat{K}(x,y)\right) = \Phi\left( \frac{\hat{P}_{X,Y}(x,y)}{\hat{P}_X(x)\hat{P}_Y(y)}\right)$, so
\begin{equation}\label{eq:generative1}
    u_{A}(\bm{\tau},P,\hat{{K}}) = \sum P_{X}P_Y\left[\Phi\left( \frac{\hat{P}_{X,Y}}{\hat{P}_X\hat{P}_Y}\right)+\left(\frac{P_{X,Y}}{P_{X}P_Y}-\frac{\hat{P}_{X,Y}}{\hat{P}_X\hat{P}_Y}\right)\hat{K}\right]
\end{equation}
On the other hand, by Definition~\ref{def:fdiv}, 
\begin{equation}\label{eq:generative2}
    D_\Phi(P_{X,Y}\| P_X P_Y) = \sum P_{X}P_Y\cdot\Phi\left( \frac{P_{X,Y}}{P_{X}P_Y}\right)
\end{equation}
By combining \eqref{eq:generative1} and \eqref{eq:generative2}, we have
\begin{align*}
    &D_\Phi(P_{X,Y}\| P_X P_Y)-u_{A}(\bm{\tau},P,{K})\\
    =& \sum P_{X}P_Y\left[\Phi\left( \frac{P_{X,Y}}{P_{X}P_Y}\right)-\Phi\left( \frac{\hat{P}_{X,Y}}{\hat{P}_X\hat{P}_Y}\right)-\left(\frac{P_{X,Y}}{P_{X}P_Y}-\frac{\hat{P}_{X,Y}}{\hat{P}_X\hat{P}_Y}\right)\hat{K}\right]\\
    \le& \sum P_{X}(x)P_Y(y)\left|\Phi\left( \frac{P_{X,Y}(x,y)}{P_{X}(x)P_Y(y)}\right)-\Phi\left( \frac{\hat{P}_{X,Y}(x,y)}{\hat{P}_X(x)\hat{P}_Y(y)}\right)\right|\\
    &+\sum P_{X}(x)P_Y(y)\left|\frac{P_{X,Y}(x,y)}{P_{X}(x)P_Y(y)}-\frac{\hat{P}_{X,Y}(x,y)}{\hat{P}_X(x)\hat{P}_Y(y)}\right|\cdot |\hat{K}(x,y)|
\end{align*}
Thus, it is sufficient to show
\begin{align}
    & \sum P_{X}(x)P_Y(y)\left|\Phi\left( \frac{P_{X,Y}(x,y)}{P_{X}(x)P_Y(y)}\right)-\Phi\left( \frac{\hat{P}_{X,Y}(x,y)}{\hat{P}_X(x)\hat{P}_Y(y)}\right)\right|\le \frac{3c_L\delta}{\alpha^2}\label{eq:tv2fdiv1}\\
    &\sum P_{X}(x)P_Y(y)\left|\frac{P_{X,Y}(x,y)}{P_{X}(x)P_Y(y)}-\frac{\hat{P}_{X,Y}(x,y)}{\hat{P}_X(x)\hat{P}_Y(y)}\right|\cdot |\hat{K}(x,y)|\le \frac{3c_L\delta}{\alpha^2}\label{eq:tv2fdiv2}
\end{align}

For all $x$ $P_X(x)$ is nonzero by assumption~\ref{ass:apriori}.  By the assumption in the statement $P_{X,Y}>2\alpha$ if it's not zero, so $P_X(x)>2\alpha$.  Furthermore, since $\|P_{X,Y}-\hat{P}_{X,Y}\|_{TV}\le \delta<\alpha$, $\hat{P}_X(x)\ge \alpha$.  Therefore for all $x$ and $y$, $P_{X,Y}(x,y)\neq 0$ we have
\begin{equation}\label{eq:tv2fdiv3}
    \alpha\le \frac{P_{X,Y}(x,y)}{P_{X}(x)P_Y(y)}\text{ and }\frac{\hat{P}_{X,Y}(x,y)}{\hat{P}_X(x)\hat{P}_Y(y)}\le \frac{1}{\alpha}.
\end{equation}
To prove \eqref{eq:tv2fdiv2}, we first show an upper bound for $|\hat{K}(x,y)|$.  By the definition of $\hat{K}$, it is in the sub-gradient of $\Phi$ at $\frac{\hat{P}_{X,Y}(x,y)}{\hat{P}_X(x)\hat{P}_Y(y)}$, and it is in $[ \alpha,1/\alpha]$ due to \eqref{eq:tv2fdiv3}.  Since $\Phi$ being $c_L$-Lipschitz in such interval, we have
\begin{equation}\label{eq:tv2fdiv4}
|\hat{K}(x,y)| \le c_L.
\end{equation}
We are ready to prove \eqref{eq:tv2fdiv2}.
\begin{align*}
    &\sum_{x,y :P_{X}(x)P_Y(y)\neq 0} P_{X}(x)P_Y(y)\left|\frac{P_{X,Y}(x,y)}{P_{X}(x)P_Y(y)}-\frac{\hat{P}_{X,Y}(x,y)}{\hat{P}_X(x)\hat{P}_Y(y)}\right||\hat{K}(x,y)|\\
    \le& \sum_{x,y :P_{X}(x)P_Y(y)\neq 0} P_{X}(x)P_Y(y)\left|\frac{P_{X,Y}(x,y)}{P_{X}(x)P_Y(y)}-\frac{\hat{P}_{X,Y}(x,y)}{\hat{P}_X(x)\hat{P}_Y(y)}\right|c_L\tag{by \eqref{eq:tv2fdiv4}}\\
    =& c_L\sum\frac{1}{\hat{P}_X(x)\hat{P}_Y(y)}\left|P_{X,Y}(x,y)\hat{P}_X(x)\hat{P}_Y(y)-\hat{P}_{X,Y}(x,y){P_{X}(x)P_Y(y)}\right|\\
    \le& \alpha^2c_L\sum \left|P_{X,Y}(x,y)\hat{P}_X(x)\hat{P}_Y(y)-\hat{P}_{X,Y}(x,y){P_{X}(x)P_Y(y)}\right|\tag{$\hat{P}_X,\hat{P}_Y\ge \alpha$}\\
    \le& \alpha^2c_L\sum P_{X,Y}\left|\hat{P}_X\hat{P}_Y-P_{X}P_Y\right|+P_{X}P_Y\left|P_{X,Y}-\hat{P}_{X,Y}\right|\\
    \le& \alpha^2c_L\sum \left|\hat{P}_X\hat{P}_Y-P_{X}P_Y\right|+\left|P_{X,Y}-\hat{P}_{X,Y}\right| \le 3\alpha^2c_L\delta.
\end{align*}

Now let's prove \eqref{eq:tv2fdiv1}.  Because $\Phi$ is $c_L$-Lipschitz in $[ \alpha, 1/\alpha]$, by \eqref{eq:tv2fdiv3}, we have
\begin{equation}\label{eq:tw2fdiv5}
    \left|\Phi\left( \frac{P_{X,Y}(x,y)}{P_{X}(x)P_Y(y)}\right)-\Phi\left( \frac{\hat{P}_{X,Y}(x,y)}{\hat{P}_X(x)\hat{P}_Y(y)}\right)\right|\le c_L\left| \frac{P_{X,Y}(x,y)}{P_{X}(x)P_Y(y)}- \frac{\hat{P}_{X,Y}(x,y)}{\hat{P}_X(x)\hat{P}_Y(y)}\right|.
\end{equation}
With argument similar to the proof of \eqref{eq:tv2fdiv2}, we completes the proof.
\end{proof}
\subsection{Proofs in Sect.~\ref{sec:algorithm}}

\begin{proof}[Proof of Lemma~\ref{lem:breg2emp}]
Because $\tilde{K}$ satisfies Eqn.~\eqref{eq:erm} and $K^\star\in\mathcal{K}$, we have
$$\int \tilde{K}d\tilde{P}_{X,Y}-\int \Phi^*(\tilde{K})d\tilde{P}_X\tilde{P}_Y\ge \int K^\star d\tilde{P}_{X,Y}-\int \Phi^*(K^\star)d\tilde{P}_X\tilde{P}_Y.$$
On the other hand,
$$B_{\Phi^*, P_X P_Y}(\tilde{K},K^\star) = \int \Phi^*(\tilde{K})- \Phi^*(K^\star)dP_X P_Y-\int\left(\tilde{K}-K^\star\right)dP_{X,Y}.$$
Combining these two we have an upper bound for $B_{\Phi^*, P_X P_Y}(\tilde{K},K^\star)$,
$$\int \left(\Phi^*(\tilde{K})- \Phi^*(K^\star)\right)\left(dP_X P_Y-d\tilde{P}_X\tilde{P}_Y\right)-\int\left(\tilde{K}-K^\star\right)\left(dP_{X,Y}-d\tilde{P}_{X,Y}\right)$$
which completes the proof.
\end{proof}

\begin{proof}[Proof of Theorem~\ref{thm:empest}]
By Lemma~\ref{lem:bregman} and \ref{lem:breg2emp}, we know the error between $D_\Phi(P_{X,Y}\| P_X P_Y)-u_{A}(\bm{\tau}, P, K)$ can be upper bound by
\begin{align}
     &\sup_{k\in\mathcal{K}}\left|\int \Phi^*(k)-\Phi^*(K^\star)d(\tilde{P}_X\tilde{P}_Y-P_X P_Y)\right|\label{eq:empest1}\\
     &\sup_{k\in\mathcal{K}}\left|\int k-K^\star d(\tilde{P}_{X,Y}-P_{X,Y}) \right|.\label{eq:empest2}
\end{align}
Now we can apply the uniform bound in Theorem~\ref{thm:uniform_emp} for \eqref{eq:empest2}.  By taking $A = \frac{\varepsilon L_1}{R_1^2}$, $B = 1$, $L = L_1$, $R = R_1$, and $\epsilon = \varepsilon\sqrt{n}$, we have
\begin{align*}
    \Pr\left[\sup_{k\in\mathcal{K}}|v_n(k)|\ge \epsilon\right] &= \Pr\left[\sup_{k\in\mathcal{K}}\left|\sqrt{n}\int k d(\hat{P}_n-P)\right|\ge \varepsilon \sqrt{n}\right]\\
    &= \Pr\left[\sup_{k\in\mathcal{K}}\left|\int k d(\hat{P}_n-P)\right|\ge \varepsilon\right]\\
    &\le  B\exp\left(-\frac{\epsilon^2}{B^2(A+1)R_1^2}\right)\\
    &\le  \exp\left(-\frac{\varepsilon^2}{(A+1)R_1^2} n\right) \le \frac{1}{n^2}\delta.
\end{align*}
The last inequality is true by taking $n = m_L/3 = O\left(\frac{(A+1)R_1^2}{\varepsilon^2}\log \frac{n^2}{\delta}\right) = O\left(\frac{1}{\varepsilon^2}\log \frac{n^2}{\delta}\right)$ when $\varepsilon$ is small enough.  We can derive similar upper bound for \eqref{eq:empest2}, and we complete the proof
\end{proof}

\subsection{Proof of Theorem~\ref{thm:nonexists}}\label{sec:proof_impossiblity}
\begin{proof}
Let's consider the following prior distribution $P_{X,Y}$: Given non-negative variables $\alpha, \beta, \gamma$ such that $\alpha+\beta+\gamma\le 1$, we set the distribution over $\X\times\Y = \{1,2,3\}\times \{1,2,3\}$ to be
$$P_{X,Y} = \frac{1}{3}
\begin{pmatrix}
1-\alpha-\beta   &\alpha &\beta \\
\alpha  & 1-\alpha-\gamma & \gamma  \\
\beta   &   \gamma  & 1-\beta-\gamma
\end{pmatrix}$$
An empirical distribution (histogram) from $x_1, \ldots, x_m$ can be represented by $9$ integers $\mathbf{m} = (m_{k,l})$ where $k$ and $l$ are between $1$ to $3$ and $m_{k,l}$ is the number of $(k,l)$ in those $m$ samples, and the distribution of $\mathbf{m}$ forms a multi-nomial distribution.  Therefore we can compute the expectation of $\hat{D}$,
\begin{align*}
    &\E[\hat{D}(\mathbf{x}, \mathbf{y})]\\
    =& \sum_{\mathbf{m}:\sum m_{k,l} = m}\frac{m!}{\prod_{k,l} m_{k,l}!}\prod_{k,l} P_{X,Y}(k,l)^{m_{k,l}}\hat{D}(\mathbf{m})\\
    =& \sum_{\mathbf{m}:\sum m_{k,l} = m}\frac{m!}{\prod_{k,l} m_{k,l}!}\left(\frac{\alpha}{3}\right)^{m_{1,2}+m_{2,1}}\left(\frac{\beta}{3}\right)^{m_{1,3}+m_{3,1}}\left(\frac{\gamma}{3}\right)^{m_{2,3}+m_{3,2}}\\
    &\cdot\left(1-\frac{\alpha+\beta}{3}\right)^{m_{1,1}}\left(1-\frac{\beta-\gamma}{3}\right)^{m_{2,2}}\left(1-\frac{\alpha+\gamma}{3}\right)^{m_{3,3}}\hat{D}(\mathbf{m})
\end{align*}
which is a polynomial of $\alpha,\beta$ and $\gamma$.  On the other hand, the $\Phi$-divergence is
\begin{align*}
    \frac{1}{9}\left[2\Phi\left(3\alpha\right)+2\Phi\left(3\beta\right)+2\Phi\left(3\gamma\right)+\Phi\left(3(1-\alpha-\beta)\right)+\Phi\left(3(1-\beta-\gamma)\right)+\Phi\left(3(1-\alpha-\gamma)\right)\right]
\end{align*}
By taking partial derivative with respect to $\alpha$ then $\beta$, 
$$\frac{\partial^2}{\partial \beta \partial \alpha}\E[\hat{D}(x_1, x_2, \ldots, x_m)] = \frac{\partial^2}{\partial \beta \partial \alpha}\E_{ P_X P_Y}\left[\Phi\left(\frac{P_{X,Y}}{P_X P_Y}\right)\right] = \Phi''(3(1-\alpha-\beta))$$
which implies the second derivative of $\Phi$ is a polynomial on $(0,3)$ and $\Phi(a)$ is a polynomial on $(0,3)$.

Similarly we take another prior distribution
$$P_{X,Y}' =
\begin{pmatrix}
1-\alpha-\beta   &0 &0 \\
0  & \alpha & 0  \\
0   &   0  & \beta
\end{pmatrix}$$
and the $\Phi$-divergence is
$$\alpha^2\Phi(1/\alpha)+b^2\Phi(1/\beta)+(1-\alpha-\beta)^2\Phi(1/(1-\alpha-\beta))+(1-\alpha^2-\beta^2-(1-\alpha-\beta)^2)\Phi(0)$$
By taking partial derivative with respect to $\alpha$ and $\beta$ we have
$$2\Phi\left(\frac{1}{x}\right)-\frac{2}{x}\Phi'\left(\frac{1}{x}\right)+\frac{1}{x^2}\Phi''\left(\frac{1}{x}\right)$$
is a polynomial with respect to $x$.

Combining these two statements we have if there are unbiased estimators for $P_{X,Y}$ and $P'_{X,Y}$, the convex function $\Phi$ is a degree one polynomial which reaches a contradiction.
\end{proof}
\section{Relation to CA mechanism and mutual information mechanism}\label{sec:comparison}
\subsection{CA mechanism}

In the following proposition we show CA mechanism is a special case of our mechanism~\ref{alg:fmechansim}.  

\begin{prop}[CA mechanism~\cite{Shnayder2016-xx}]
If we take $\Phi(a) = \frac{1}{2}|a-1|$ and restrict $|K|\le 1/2$, Then the above mechanism reduces to the Correlated Agreement mechanism.
\end{prop}
However the corresponding $\Phi$ is not strictly convex, so the mechanism is not strongly truthful in general.  \begin{proof}
If we take $\Phi(a) = \frac{1}{2}|a-1|$, $\Phi^*(b) = b$ when $|b|\le 1/2$, the payment can be simplified as
$$M_{A}(\mathbf{r}) = K\left(\hat{x}_{b}, \hat{y}_{b}\right)-\Phi^*\left(K\left(\hat{x}_{p}, \hat{y}_{q}\right)\right)= K\left(\hat{x}_{b}, \hat{y}_{b}\right)-K\left(\hat{x}_{p}, \hat{y}_{q}\right).$$
Moreover, by Table~\ref{tab:fconjugate}, Eqn.~\eqref{eq:optimal} reduces to $$\partial \Phi\left(\frac{P_{X,Y}(x,y)}{P_{X}(x)P_{Y}(y)}\right) = \begin{cases} 1/2 &\text{ if } P_{X,Y}(x,y)>P_X(x)P_Y(y);\\
-1/2 &\text{ if } P_{X,Y}(x,y)<P_X(x)P_Y(y);\\
[-1/2,1/2] &\text{ otherwise,}
\end{cases}$$
and the scoring functions are in $\partial \Phi\left(\frac{P_{X,Y}(x,y)}{P_{X}(x)P_{Y}(y)}\right)+\frac{1}{2}$.
\end{proof}
\subsection{Mutual Information Mechanism}
The framework of mutual information mechanism~\cite{kong2019information} defines the payments to Alice (and Bob) to be the $\Phi$-mutual information between Alice's and Bob's reports
$$D_\Phi(\bm{\theta}\circ P_{X,Y}\| \bm{\theta}\circ P_X P_Y)$$
where $P$ is the joint distribution of signal and $\bm{\theta}$ is the strategy profile.  

Our mechanism can be seen as a special case of the mutual information mechanism when the number of tasks goes to infinity.  Formally, the fantasy mapping $F^\Phi$ pays Alice and Bob with the $\Phi$-mutual information between Alice's and Bob's reports.  As the number of tasks goes to infinity, both Algorithm~\ref{alg:generative} or \ref{alg:empest} are $(0,0)$-accurate and by the proof of Theorem~\ref{thm:framework}, our mechanism with learning Algorithm~\ref{alg:generative} or \ref{alg:empest} pays Alice with $F^\Phi_A$ which is the $\Phi$-mutual information between Alice's and Bob's reports.

However, our mechanism has a stronger guarantee when the number of tasks is finite.  In the proof of Theorem~\ref{thm:framework}, our mechanism ensure Alice's ex-ante payment is upper bounded by the $\Phi$-mutual information between Alice's and Bob's reports \emph{uniformly under any strategy profiles}.  This property does not hold if we estimate the $\Phi$-mutual information directly without variational representation.  

For example, in \citet{kong2019information}, they use the agents' report profile to estimate the density function and estimate the $\Phi$-mutual information between their reports directly.  In contrast, although Mechanism~\ref{alg:reduction} with learning Algorithm~\ref{alg:generative} also first estimates the density function, the mechanism then computes a scoring function instead. These two methods have similar behavior under the truth-telling strategy profile.  However, given a fixed the number of tasks, there may exist a non-truthful strategy profile $\bm{\theta}$ such that we cannot estimate the density function $\bm{\theta}\circ P$ accurately.  The method in \citet{kong2019information} cannot provide a guarantee in such a situation.  On the other hand, our variational method ensures the ex-ante payments to agents under $\bm{\theta}$ are worse than the mutual information between agents' reports.  Having a uniform upper bound for a non-truthful strategy is important for real application.  We may assume our learning algorithm can estimate agents' signal distributions, which is derived from non-adversarial settings.  However, agents adopt the worst possible strategy profiles  to break our mechanism adversarially.
\end{document}